\documentclass[conference]{IEEEtran}
\usepackage{xcolor}
\usepackage[utf8]{inputenc}
\usepackage{amsmath}
\usepackage{amsfonts}
\usepackage{mathtools}
\usepackage{tabularx}
\usepackage{booktabs}
\usepackage{xspace}
\usepackage{boxedminipage}
\usepackage[inline]{enumitem}
\usepackage{amsthm}
\usepackage{multirow}
\usepackage{pifont}
\usepackage{makecell}
\usepackage{adjustbox}
\usepackage[hyphens]{url}
\usepackage[hidelinks]{hyperref}
\usepackage{cleveref}

\ifCLASSOPTIONcompsoc
    \usepackage[caption=false, font=normalsize, labelfont=sf, textfont=sf]{subfig}
\else
\usepackage[caption=false, font=footnotesize]{subfig}
\fi

\renewcommand{\paragraph}[1]{\noindent \textbf{#1\xspace}}

\newcommand{\sig}[1]{\langle #1 \rangle}
\newcommand{\floor}[1]{\lfloor #1 \rfloor}

\definecolor{yescolor}{HTML}{026378}
\def\yes{\textcolor{yescolor}{\checkmark}}
\def\no{\textcolor{red}{\ding{55}}}

\newtheorem{definition}{Definition}
\newtheorem{corollary}{Corollary}
\newtheorem{lemma}{Lemma}
\newtheorem{theorem}{Theorem}
\newtheorem{claim}{Claim}
\newtheorem{fact}{Fact}

\newcommand{\Propose}{\mathsf{propose}}
\newcommand{\OptPropose}{{\sf opt\text{-}propose}}
\newcommand{\FallbackPropose}{{\sf fb\text{-}propose}}
\newcommand{\Timeout}{\mathsf{timeout}}
\newcommand{\Vote}{\mathsf{vote}}
\newcommand{\OptVote}{\mathsf{opt\text{-}vote}}
\newcommand{\FallbackVote}{\mathsf{fb\text{-}vote}}
\newcommand{\Cert}{\mathcal{C}}

\newcommand{\TimeoutMessage}{\mathcal{T}}
\newcommand{\TimeoutCert}{\mathcal{TC}}
\newcommand{\ViewTimer}{{\sf view\text{-}timer}}
\newcommand{\network}{\mathcal{V}\xspace}
\newcommand{\node}[1]{\ensuremath{P_{#1}}\xspace}
\newcommand{\protocol}{\mathcal{P}\xspace}
\newcommand{\Lock}{{\sf lock}}
\newcommand{\Status}{\mathsf{status}}
\newcommand{\TimeoutView}{{\sf timeout\_view}}
\newcommand{\Commit}{\mathsf{commit}}

\newcommand{\mvbp}{$\omega$\xspace}
\newcommand{\mcl}{$\lambda$\xspace}
\newcommand{\vl}{$\tau$\xspace}

\newcommand{\moonshot}{Moonshot\xspace}
\newcommand{\MoonshotShort}{M\xspace}
\newcommand{\simple}{Simple \moonshot}
\newcommand{\pipelined}{Pipelined \moonshot}
\newcommand{\commit}{Commit \moonshot}
\newcommand{\SimpleShort}{S\MoonshotShort}
\newcommand{\PipelinedShort}{P\MoonshotShort}
\newcommand{\CommitShort}{C\MoonshotShort}
\newcommand{\Best}{$\mathcal{B}$\xspace}
\newcommand{\WorstJolteon}{$\mathcal{WJ}$\xspace}
\newcommand{\WorstMoonshot}{$\mathcal{WM}$\xspace}

\title{\moonshot: Optimizing Chain-Based Rotating Leader BFT via Optimistic Proposals}
\date{}
\author{\IEEEauthorblockN{Isaac Doidge, Raghavendra Ramesh, Nibesh Shrestha, and Joshua Tobkin}
\IEEEauthorblockA{Supra Research\\
\{i.doidge, r.ramesh, n.shrestha, j.tobkin\}@supraoracles.com}}

\begin{document}

\maketitle

\begin{abstract}
Existing chain-based rotating-leader BFT SMR protocols for the partially synchronous network model with constant commit latencies incur block periods of at least $2\delta$ (where $\delta$ is the message transmission latency). While a protocol with a block period of $\delta$ exists under the synchronous model, its commit latency is linear in the size of the system.

To close this gap, we present the first chain-based BFT SMR protocols with $\delta$ delay between the proposals of consecutive honest leaders and commit latencies of $3\delta$. We present three protocols for the partially synchronous model under different notions of optimistic responsiveness, two of which implement pipelining. All of our protocols achieve reorg resilience and two have short view lengths; properties that many existing chain-based BFT SMR protocols lack. We present an evaluation of our protocols in a wide-area network wherein they demonstrate significant increases in throughput and reductions in latency compared to the state-of-the-art, Jolteon. Our results also demonstrate that techniques commonly employed to reduce communication complexity---such as vote-pipelining and the use of designated vote-aggregators---actually reduce practical performance in many settings.
\end{abstract}

\section{Introduction}\label{sec:intro}
Blockchain networks have become increasingly popular as mechanisms for facilitating decentralised, immutable and verifiable computation and storage. These networks leverage Byzantine fault-tolerant (BFT) consensus protocols to ensure that their participants (called \emph{nodes}) execute the same sequence of operations (called \emph{transactions}), despite some of them exhibiting arbitrary failures. Many blockchain networks also prioritize \emph{fairness}; i.e., they strive to ensure that i) client transactions are processed promptly, without granting any client an unfair advantage over the others, and ii) nodes have an equal opportunity to be rewarded for the work that they do in the system. Public blockchain networks in particular also tend to be large, supporting hundreds (e.g. \cite{aptos-validators}) or thousands (e.g. \cite{eth-validators}) of nodes in the pursuit of decentralization, and aim to cater to many concurrent clients. Accordingly, the consensus protocols driving these networks need to be efficient, maximising transaction throughput and minimising end-to-end commit latency (i.e., the time between a client submitting a transaction and it being executed by the blockchain).

To these ends, prior works~\cite{shi2019streamlined,yin2019hotstuff,gelashvili2022jolteon,abraham2022optimal,camenisch2022internet,giridharan2023beegees,malkhi2023hotstuff} have leveraged two key strategies: i) \emph{block chaining}, and; ii) frequent \emph{leader rotation}. In the block chaining (or \emph{chained}) paradigm, transactions are grouped into \emph{blocks} that explicitly reference one or more existing blocks (called the \emph{parents} of the block), typically by including their hashes. This enables an optimization called \emph{pipelining}, wherein the \emph{vote} acknowledgement messages sent by the nodes in the course of agreeing upon a block can be counted towards the finalization of its parents, facilitating the removal of additional voting phases and thus reducing the communication and computational complexity of the protocol by a constant factor. Our work focuses on the \emph{chain-based} subcategory of chained protocols, wherein each block has exactly one parent, as opposed to DAG-based protocols in which a block may have many parents. In rotating-leader chain-based protocols, the leader responsible for proposing these blocks is changed at regular intervals, even when functioning correctly. This helps to fairly distribute the proposal workload and any related rewards. Additionally, the more frequently leaders are rotated the less amount of time a Byzantine (faulty) leader has to manipulate the ordering of pending transactions, improving censorship resistance. Accordingly, rotating-leader protocols often rotate the leader after every block proposal, an approach called \emph{leader-speaks-once} (LSO). This paper seeks to optimize chain-based BFT consensus performance in a modified version of the LSO setting, which we name \emph{leader-certifies-one} (LCO). Whereas an LSO protocol allows a leader to propose only a single block, an LCO protocol allows it to propose multiple but ensures that it produces no more than one certified block. Even as the previously cited works need not be implemented as LSO, our protocols need not be implemented as LCO, however, it is in this setting that they have the greatest advantage. We henceforth refer to chain-based BFT consensus protocols that implement leader rotation as \emph{CRL protocols}.

\begin{table*}[t]
    \scriptsize
    \centering
    \caption{\textbf{Theoretical comparison of chain-based rotating leader BFT SMR protocols}}
    \def\arraystretch{1}
    \renewcommand{\thefootnote}{\fnsymbol{footnote}}
    \setlength\tabcolsep{1mm}
    \begin{center}
    \begin{tabularx}{\textwidth}{lr c c c c c c c c c c c}
    \toprule
    \multicolumn{2}{c}{} &
    {\textbf{Model}} &
    \multirow{2}{*}{\makecell{\textbf{Minimum Commit} \\ \textbf{Latency}}} &
    \multirow{2}{*}{\makecell{\textbf{Minimum View Change} \\ \textbf{Block Period}}} &
    \multirow{2}{*}{\makecell{\textbf{Reorg} \\ \textbf{Resilience}}} &
    \multirow{2}{*}{\makecell{\textbf{View} \\ \textbf{Length}}} &
    {\textbf{Pipelined}} &
    \multicolumn{2}{c}{\textbf{Communication Complexity$^{(1)}$}}  &
    \multicolumn{2}{c}{\textbf{Optimistic Responsiveness}}  &
    \\
    \cmidrule(lr){9-10} \cmidrule(lr){11-12}
    & & &
    & & & & &
    {\textbf{steady-state}} &
    {\textbf{view-change}} &
    {\textbf{standard}} &
    {\textbf{consecutive honest}} &
    \\
    \midrule
    HotStuff &
    {\cite{yin2019hotstuff}} & 
    {psync.} &
    {$7\delta^{(2)}$} &
    {$2\delta$} &
    {\no} & 
    $4\Delta$ &
     \yes  &
    $O(n)$ & 
    $O(n)$ &
    $\yes$  &
    $\yes$  &
    \\
    Fast HotStuff & 
    {\cite{jalalzai2020fast}} &
     {psync.} &
    {$5\delta$} &
    {$2\delta$} &
    {\no} &
    $4\Delta$ &
    $\yes$  &
    $O(n)$ & 
    $O(n^2)$ &
    $\yes$  &
    $\yes$  &
    \\
     Jolteon &
    {\cite{gelashvili2022jolteon}} & 
     {psync.} &
     {$5\delta$} &
    {$2\delta$} &
    {\no} &
    $4\Delta$ &
     $\yes$  &
    $O(n)$ &
    $O(n^2)$ &
    $\yes$ &
    $\yes$  &
    \\
    HotStuff-2 &
    {\cite{malkhi2023hotstuff}} & 
     {psync.} &
     {$5\delta$} &
    {$2\delta$} &
    {\yes} &
    {$7\Delta$} &
     $\yes$  &
    $O(n)$ &
    $O(n)$ &
    \no &
    $\yes$  &
    \\
    PaLa &
    {\cite{chan2018pala}} &
     {psync.} &
    {$4\delta$} &
    {$2\delta$} &
    {\yes} &
    {$5\Delta$} &
     $\yes$  &
    $O(n^2)$ &
    $O(n^2)$ &
    \no &
    \yes  &
    \\
    ICC &
    {\cite{camenisch2022internet}} &
     {psync.} &
    {$3\delta$} &
    {$2\delta$} &
    \no  &
    {$4\Delta$} &
     \no  &
    $O(n^2)$ &
    $O(n^2)$ &
    \no &
    \yes  &
    \\
    Simplex &
    {\cite{chan2023simplex}} &
     {psync.} &
    {$3\delta$} &
    {$2\delta$} &
    \yes  &
    {$3\Delta$} &
      \no  &
    {Unbounded}$^{(3)}$ &
    $O(n^2)$ &
    \no$^{(4)}$&
    \yes &
    \\
    \midrule
    Apollo &
    {\cite{bhat2023unique}} &
     {sync.} &
    {$(f+1)\delta$} &
    {$\delta$} &
    \yes  &
    {$4\Delta$} &
      \no  &
    {$O(n)$} &
    $O(n^2)$ &
    \no  &
    \yes  &
    \\
    \midrule
    \multicolumn{2}{l}{\textbf{This work ($\S\ref{sec:moonshot_v1}$)}} & 
      {psync.} &
    {$3\delta$} &
    {$\delta$} &
    \yes  &
    {$5\Delta$} &
      \yes  &
    $O(n^2)$ & 
    $O(n^2)$ & 
    \no & 
    $\yes$  &
    \\
    \multicolumn{2}{l}{\textbf{This work ($\S\ref{sec:moonshot_v2}$)}} & 
      {psync.} &
    {$3\delta$} &
    {$\delta$} &
    \yes  &
    {$3\Delta$} &
     \yes  &
    $O(n^2)$ & 
    $O(n^2)$ &
    $\yes$ & 
    $\yes$  &
    \\
    \multicolumn{2}{l}{\textbf{This work ($\S\ref{sec:commit-moonshot}$)}} & 
      {psync.} &
    {$3\delta$} &
    {$\delta$} &
    \yes  &
    {$3\Delta$} &
     \no  &
    $O(n^2)$ & 
    $O(n^2)$ &
   
    $\yes$ & 
    $\yes$  &
    \\
    \bottomrule
    \end{tabularx}
    \end{center}
    (1) Assuming the use of threshold signatures.
    (2) HotStuff has a minimum commit latency of $7\delta$ if the next leader aggregates the votes for the current leader's proposal. In the original HotStuff specification, leaders aggregate the votes for their own proposals and then forward the resultant QC to the next leader, incurring an additional $3\delta$.
    (3) Simplex~\cite{chan2023simplex} requires each proposal to include its notarized parent blockchain, making the size of each proposal proportional to the size of the blockchain itself.
    (4) Simplex~\cite{chan2023simplex} claims responsiveness only when all nodes are honest.
    \label{tbl:related_work}
\end{table*}

Our work targets the \emph{partially synchronous} network model~\cite{dwork1988consensus} wherein there exists a time called the \emph{Global Stabilization Time} (GST) after which message delivery takes at most $\Delta$ time. We use $\delta$ to denote the actual delivery time, which naturally satisfies $\delta \le \Delta$ after GST. Many recent CRL protocols for this setting have focused on reducing communication complexity. Some have achieved linear communication complexity in their \textit{steady state} phases~\cite{jalalzai2020fast,gelashvili2022jolteon} (i.e. when the protocol makes progress under a fixed leader), while others obtain this result in their \emph{view-change} phases~\cite{yin2019hotstuff,malkhi2023hotstuff} (i.e. when the protocol elects a new leader) as well. However, these protocols sacrifice efficiency in several important metrics in their pursuit of linearity, including i) \emph{minimum commit latency} (i.e., the minimum delay between a block being proposed and it being committed by all honest---i.e., non-faulty---nodes), ii) \emph{minimum view change block period} (i.e., the minimum delay between the proposals of different honest leaders), and iii) \emph{view length} (i.e., the duration a node waits in a view before it considers the current leader to have failed). In particular, these works require at least $5\delta$ to commit a new block, at least $2\delta$ between honest proposals in the LSO setting, and view lengths of at least $4\Delta$. Moreover, since these protocols all rely on a designated node to aggregate vote messages and forward the resulting certificates, they grant the adversary the power to censor certificates for honest proposals when this aggregator is Byzantine---even after GST. Accordingly, any implementation of these protocols that uses any node other than the original proposer as the vote aggregator is not \emph{reorg resilient}; i.e., it cannot guarantee that an honest leader that proposes after GST will produce a block that becomes a part of the committed blockchain.

A recent line of work~\cite{camenisch2022internet, chan2023simplex} designed CRL protocols with minimum commit latencies of $3\delta$. However, these protocols are in the non-pipelined setting, have minimum view change block periods of $2\delta$ and either have long view lengths~\cite{camenisch2022internet} or are less practical in nature~\cite{chan2023simplex}. To the best of our knowledge, Apollo~\cite{bhat2023unique} is the only existing CRL protocol with a minimum view change block period of $\delta$. However, it incurs a minimum commit latency of $(f+1)\delta$ even during failure-free executions and assumes a synchronous network. As far as we know, no chain-based consensus protocol has simultaneously achieved a minimum view change block period of $\delta$ and a constant commit latency. To close this gap, our paper explores the design of such protocols, which we collectively refer to as \emph{Moonshot} protocols.

\subsection{Contributions}
\paragraph{Pipelined \moonshot protocols.}
We first present two state machine replication (SMR) protocols for the pipelined setting, each of which satisfies a different notion of \emph{responsiveness}: i) \emph{optimistic responsiveness}~\cite{yin2019hotstuff} (\Cref{def:responsiveness}) and ii) \emph{optimistic responsiveness under consecutive honest leaders}~\cite{giridharan2023beegees} (\Cref{def:responsiveness_ch}). Informally, the former requires an honest leader to make progress in $O(\delta)$ time after GST (i.e., without waiting for $\Omega(\Delta)$ time) while the latter requires an honest leader to make progress in $O(\delta)$ time only when the previous leader is also honest. Our first protocol satisfies the former definition and is simpler to reason about, but has a longer view length. The second satisfies the latter definition and has a shorter view length, but is more complex.

Both of our protocols require only two consecutive honest leaders after GST to commit a new block, and achieve reorg resilience through vote-multicasting. This strategy, together with an optimization that we call \emph{optimistic proposal}, also enables them to achieve both a minimum view change block period of $\delta$ and a minimum commit latency of $3\delta$. We say that a protocol implements optimistic proposal if a leader is allowed to ``optimistically'' extend a block proposed by its predecessor without waiting to observe its certification. We implement this in our protocols by allowing the leader of the next view to propose a new block when it votes for a block made by the leader of the current view.

\paragraph{Non-pipelined \moonshot protocols.}
As mentioned, pipelining can reduce the communication and computational overhead of a protocol. However, while this gives pipelined protocols good latency when all messages require a similar amount of time to propagate and process, pipelining actually increases commit latency when blocks take sufficiently longer to propagate or process than votes. Accordingly, we also present a non-pipelined variant of our second protocol in~\Cref{sec:commit-moonshot}. This final protocol retains standard optimistic responsiveness and requires only a single honest leader to commit a new block after GST.

\paragraph{Evaluation.} Subsequently, we present an evaluation of LCO implementations of our protocols against an LSO implementation of Jolteon. Our protocols outperformed Jolteon in failure-free wide-area networks (WANs) of up to $200$ nodes, committing approximately $1.5$x as many blocks at around half the latency, on average. Our protocols also outperformed under failures, with our non-pipelined protocol committing $8$x as many blocks with a reduction in latency by more than two orders of magnitude under Jolteon's worst-case leader schedule.

\paragraph{Organization.} The rest of the paper is organized as follows: In~\Cref{sec:preliminaries}, we present the system model and preliminaries for our work.  \Cref{sec:moonshot_v1} presents a pipelined CRL protocol with a minimum commit latency of $3\delta$, minimum view change block period of $\delta$, reorg resilience and optimistic responsiveness under consecutive honest leaders. We then modify this protocol in~\Cref{sec:moonshot_v2} to obtain a protocol with standard optimistic responsiveness and improved view length. In~\Cref{sec:commit-moonshot}, we give a non-pipelined version of our second protocol, which offers improved commit latency when blocks take sufficiently longer to propagate or process than votes. We present an evaluation of our protocols in~\Cref{sec:eval}, and conclude with a more detailed discussion of related works in~\Cref{sec:related_work}.

\section{Preliminaries}\label{sec:preliminaries}
We consider a system comprised of a set $\network = (\node{1}, \ldots, \node{n})$ of $n$ nodes running a protocol $\protocol$ in a reliable, authenticated all-to-all network. We assume the existence of a static, computationally bounded adversary that cannot break cryptographic primitives but may corrupt up to $f < n/3$ of the nodes when $\protocol$ begins, which it may then cause to behave arbitrarily. We refer to all nodes under the control of the adversary as being \emph{Byzantine}, while we refer to those that adhere to $\protocol$ as being \emph{honest}. We define a \emph{quorum} as a set of $\floor{\frac{n}{2}} + f + 1$ nodes. Henceforth, for the sake of simplicity, we assume that $n = 3f+1$ and that a quorum therefore contains $2f+1$ nodes.

We assume that each node has access to a local clock and that these clocks collectively have \emph{no drift} and \emph{arbitrary skew}. We also assume the \emph{partially synchronous communication} model of Dwork et al.~\cite{dwork1988consensus}. Under this model, the network starts in an initial state of asynchrony during which the adversary may arbitrarily delay messages sent by honest nodes. However, after an unknown time called the \emph{Global Stabilization Time} (GST), the adversary must ensure that all messages exchanged between honest nodes are delivered within $\Delta$ time of being sent (from the perspective of the sender). In our initial analyses, we denote the range of the actual transmission latencies of messages of all types with $\delta$, and observe that $\delta = [0, \Delta]$ after GST. Moreover, when we measure latency in terms of $\delta$, e.g. $x = y\delta$, we are denoting that $x$ requires the propagation of $y$ sequential messages (i.e. $x$ requires $y$ network hops). In our later analyses we base our communication model on the \emph{modified partially synchronous model} of Blum et al.~\cite{blum2023analyzing}. Under this model, we denote the range of the actual delivery times of small messages (such as votes) with $\rho$ and that of large messages (such as block proposals) with $\beta$, such that $\rho = [0, \min(\beta))$ and $\beta = (\max(\rho), \Delta]$, after GST. We follow a similar convention as with $\delta$ when measuring latency in terms of $\beta$ and $\rho$, with $x = y\beta + z\rho$ denoting that $x$ requires the sequential propagation of $y$ large and $z$ small messages.

We make use of digital signatures and a public-key infrastructure (PKI) to prevent spoofing and replay attacks and to validate messages. We use $\sig{x}_i$ to denote a message $x$ digitally signed by node $\node{i}$ using its private key. In addition, we use $\sig{x}$ to denote an unsigned message $x$ sent via an authenticated channel. We use $H(x)$ to denote the invocation of the hash function $H$ with input $x$.

\subsection{Property Definitions}

\paragraph{State Machine Replication.}
A state machine replication (SMR) protocol run by a network $\network$ of $n$ nodes receives requests (transactions) from external parties, called \textit{clients}, as input, and outputs a totally ordered log of these requests. We recall the definition of SMR given in~\cite{abraham2020sync}, below.

\begin{definition}[Byzantine Fault-Tolerant State Machine Replication~\cite{abraham2020sync}]
\label{def:smr}
A Byzantine fault-tolerant state machine replication protocol commits client requests as a linearizable log to provide a consistent view of the log akin to a single non-faulty node, providing the following two guarantees.
   \begin{itemize}[leftmargin=*]
   \itemsep0em 
        \item \textbf{Safety}. Honest nodes do not commit different values at the same log position.
        \item \textbf{Liveness}. Each client request is eventually committed by all honest nodes.
   \end{itemize}
\end{definition}

We clarify that the liveness SMR property only applies to transactions that are received by honest nodes. The protocols that we present in this paper guarantee that honest nodes continue to add new blocks proposed by honest leaders to their local blockchains. Therefore, they satisfy SMR liveness as long as their implementations ensure that transactions that are included in or referenced by failed blocks are resubmitted by honest leaders until they are included in a block that becomes committed. They also guarantee that if any two honest nodes commit a block at the same position in their local blockchains, then they commit the same block. Accordingly, they therefore satisfy SMR safety assuming that their implementations use a deterministic function that is consistent across all nodes to commit transactions to their transaction logs. These assumptions make our protocols agnostic to the manner in which transactions are distributed throughout the network, enabling optimizations like \emph{transaction batching}~\cite{danezis2022narwhal}.

\begin{definition}[Minimum View Change Block Period (\mvbp)]\label{def:minimum_view_change_block_period}
The minimum view change block period $\omega$ of a chained consensus protocol $\protocol$ is the minimum latency between the proposal of a block $B$ by an honest node $P_i$ and its extension (directly or indirectly) by any honest node $P_j$ such that $P_j \ne P_i$.
\end{definition}

\begin{definition}[Minimum Commit Latency (\mcl)]
\label{def:minimum_commit_latency}
A consensus protocol has a minimum commit latency of $\lambda$ if all honest nodes that commit a block proposed at time $t$, do so no earlier than $t + \lambda$.
\end{definition}

In this paper, we measure the above two metrics in relation to message transmission latency and assume that message processing time is relatively negligible.

\begin{definition}[View Length (\vl)]
\label{def:view_length}
A consensus protocol has a view length of $\tau$ if an honest node that enters view $v$ at time $t$ considers the view to have failed if it remains in $v$ until $t + \tau$.
\end{definition}

\begin{definition}[Reorg Resilience]\label{def:reorg_resilience}
We say that a consensus protocol is reorg resilient if it ensures that when an honest leader proposes after GST, one of its proposals becomes certified and this proposal is extended by every subsequently certified proposal.
\end{definition}

\paragraph{Optimistic Responsiveness.} Responsiveness requires a consensus protocol to make progress in time proportional to the actual network delay ($\delta$) and independent of any known upper bound delay ($\Delta$) when a leader is honest~\cite{pass2017hybrid}. Optimistic responsiveness requires this same guarantee, but only when certain optimistic conditions hold. Several variations~\cite{yin2019hotstuff,abraham2022optimal,giridharan2023beegees,chan2023simplex} have been formulated in the literature, two of which we make use of in this paper and recall below. 

\begin{definition}[Optimistic Responsiveness~\cite{yin2019hotstuff}]\label{def:responsiveness}
After GST, any correct leader, once designated, needs to wait just for the first $n-f$ responses to guarantee that it can create a proposal that will make progress. This includes the case where a leader is replaced.
\end{definition}

We note that in~\cite{yin2019hotstuff}, the term ``make progress'' means that all honest nodes will vote for the correct (honest) leader's proposal, not that all honest nodes observe a certificate for the included block; i.e. optimistic responsiveness does not imply reorg resilience. We also clarify that for LSO/LCO protocols, these (at most) $n-f$ responses should be messages from the previous view.

\begin{definition}[Optimistic Responsiveness (Consecutive Honest)~\cite{abraham2022optimal}]\label{def:responsiveness_ch}
We say that a protocol is optimistically responsive (consecutive honest) if after GST, for any two consecutive honest leaders $L_v$ and $L_{v+1}$, $L_{v+1}$ sends its proposal within $O(\delta)$ time of receiving $L_v$'s proposal.
\end{definition}

Importantly, this variant of optimistic responsiveness allows the protocol to wait for $\Omega(\Delta)$ time before proposing in the new view when the leader of the previous view is Byzantine.

\subsection{Protocol Definitions}
We now establish some general definitions that we make use of in all of our protocols.

\paragraph{View-based execution.}
Our protocols progress through a sequence of numbered \emph{views}, with all nodes starting in view $1$ and progressing to higher views as the protocol continues. Each view $v$ is coordinated by a designated leader node $L_v$ that is responsible for proposing a new block for addition to the blockchain. For the sake of liveness, we require that the leader election function $L$ continually elects sequences of leaders that contain at least two consecutive (not necessarily distinct) honest leaders after GST for our pipelined protocols, and only one such leader for our non-pipelined protocol. We note that $L$ must additionally change the leader every view for LCO implementations, and must elect each node with equal probability in fair implementations.

\paragraph{Blocks.} The blockchains of each of our protocols are initialized with a \emph{genesis block} $B_0$ that is known to all nodes at the beginning of the protocol. Each block references its immediate predecessor in the chain, which we refer to as its \emph{parent}, with the parent of the genesis block being $\bot$. We say that a block \emph{directly extends} its parent and \emph{indirectly extends} its other predecessors in the chain. For simplicity when reasoning, we also say that a block extends itself. We refer to the predecessors of a given block as its \emph{ancestors} and measure its \emph{height} by counting its ancestors. A block $B_k$ with height $k$ has the format, $B_k := (b_v, H(B_{k-1}))$ where $b_v$ is a fixed payload for the view $v$ for which $B_k$ is proposed, $B_{k-1}$ is the parent of $B_k$, and $H(B_{k-1})$ is the hash digest of $B_{k-1}$. We allow the implementation to dictate the contents of $b_v$ (e.g. transactions or hashes of batches of transactions). Accordingly, $B_k$ is \textit{valid} if i) its parent is valid, or if $k=0$ and its parent is $\bot$, and ii) $b_v$ satisfies the implementation-specific validity conditions. Finally, we say that two blocks $B_k$ and $B'_{k'}$ proposed for the same view \textit{equivocate} one another if they do not both have the same parent and payload.

\paragraph{Block certificates.} In our protocols, a node sends a signed $\Vote$ message to indicate its acceptance of a block. A block certificate $\Cert_v(B_k)$ for view $v$ consists of a quorum of distinct signed $\Vote$ messages for $B_k$ for $v$. We use $\Cert_v$ to denote a block certificate for view $v$ when knowledge of the related block is irrelevant to the context. We rank block certificates by their views such that $\Cert_v \le \Cert_{v'}$ if $v \le v'$. We provide more detailed definitions in the following sections where necessary.

\paragraph{Timeout messages and timeout certificates.} 
Our protocols maintain the liveness SMR property by requiring nodes to request a new leader when they fail to observe progress in their current views after a certain amount of time. They do so by sending signed timeout messages for the view, the contents of which are protocol-specific. A view $v$ timeout certificate, denoted $\TimeoutCert_{v}$, consists of a quorum of distinct signed timeout messages for $v$, denoted $\TimeoutMessage_v$.

\section{\simple}\label{sec:moonshot_v1}
We now present \simple (\Cref{fig:simple-moonshot}), the first of our CRL protocols for the pipelined setting. \simple achieves \mvbp $= \delta$, \mcl $= 3\delta$, reorg-resilience and responsiveness under consecutive honest leaders. We first discuss how our protocols obtain the former properties before elaborating on \simple itself.

\begin{figure*}[!htbp]
\small
    \begin{boxedminipage}[t]{\textwidth}
    
    A \simple node $\node{i}$ runs the following protocol whilst in view $v$:
    \begin{enumerate}[leftmargin=*]
        \item \textbf{Propose.} If $\node{i}$ is $L_v$ and enters $v$ at time $t_i$, propose: (i) upon receiving $\Cert_{v-1}(B_{k-1})$ before $t_i + 2\Delta$, or; (ii) at $t_i + 2\Delta$. Do so by multicasting $\sig{\Propose, B_k, \Cert_{v'}(B_{k-1}), v}$, where $\Cert_{v'}(B_{k-1})$ is the highest ranked block certificate known to $L_v$ and $B_k$ extends $B_{k-1}$.
        
        \item \textbf{Vote.}\label{step:vote} $\node{i}$ votes once using one of the following rules:
        \begin{enumerate}[leftmargin=*]
            \item Upon receiving $\sig{\OptPropose, B_{k}, v}$ such that $B_k$ extends $B_{k-1}$, if $\Lock_i = \Cert_{v-1}(B_{k-1})$ then multicast $\sig{\Vote, H(B_k), v}_i$.
            
            \item Upon receiving $\sig{\Propose, B_{k}, \Cert_{v'}(B_{h}), v}$, if $\Cert_{v'}(B_{h}) \ge \Lock_i$ and $B_k$ extends $B_h$ then multicast $\sig{\Vote, H(B_k), v}_i$.
        \end{enumerate}
    
        \item \textbf{Optimistic Propose.} Upon voting for $B_{k}$ in $v$, if $\node{i}$ is $L_{v+1}$, multicast $\sig{\OptPropose, B_{k+1}, v+1}$ such that $B_{k+1}$ extends $B_{k}$.
        
        \item \textbf{Timeout.} Upon receiving $f+1$ distinct $\sig{\Timeout, v}_*$  or when $\ViewTimer_i$ expires, stop voting in $v$ and multicast $\sig{\Timeout, v}_{i}$.
        
        \item \textbf{Advance View.} Upon receiving $\Cert_{v'-1}(B_h)$ or $\TimeoutCert_{v'-1}$, where $v' > v$, and before executing any other rule, do the following: i) multicast the certificate; ii) update $\Lock_i$ to the highest ranked block certificate received so far; iii) unicast a status message $\sig{\Status, v', \Lock_i}$ to $L_{v'}$ if $\Lock_i$ has a view less than $v'-1$, iv) enter $v'$, and; v) reset $\ViewTimer_i$ to $5\Delta$ and start counting down.
    \end{enumerate}

    $\node{i}$ additionally performs the following action in any view:
    \begin{enumerate}[leftmargin=*]
        \item \textbf{Direct Commit.}\label{step:direct-commit} Upon receiving $\Cert_{v-1}(B_{k-1})$ and $\Cert_{v}(B_{k})$ such that $B_{k}$ extends $B_{k-1}$, commit $B_{k-1}$. 
        \item \textbf{Indirect Commit.}\label{step:indirect-commit} Upon directly committing $B_{k-1}$, commit all of its uncommitted ancestors. 
    \end{enumerate}
    \end{boxedminipage}
    \caption{The \simple Protocol}
    \label{fig:simple-moonshot}
\end{figure*}

\paragraph{Towards achieving \mvbp $= \delta$ and \mcl $= 3\delta$.}
Prior CRL protocols require $L_v$ to observe $C_{v-1}$ before proposing during their happy paths (i.e. when views progress without any honest node sending a timeout message---as opposed to the \emph{fallback path}). This is intended to help honest leaders create blocks that will become committed, but is unnecessarily strict for this purpose and naturally affects \mvbp $\ge 2\delta$ and \mcl $\ge 4\delta$ in the pipelined setting. Our protocols improve upon these results by requiring i) the leader of view $v$ to propose a block for $v$, say $B_k$, upon voting for a block in $v-1$, say $B_{k-1}$, and; ii) nodes to multicast their votes. Allowing leaders to propose optimistically in this way enables voting for $B_{k-1}$ to proceed in parallel with the proposal of $B_k$. Moreover, when the dissemination times of vote and proposal messages are equal (see~\Cref{fig:pipelined-moonshot-communication}), having nodes multicast their votes ensures that if all honest nodes vote for $B_{k-1}$ then they will all receive $B_k$ at the time that they construct $\Cert_{v-1}(B_{k-1})$, allowing them to vote for $B_k$ and $L_{v+1}$ to propose immediately upon entering $v$. Hence, in the happy path, $L_{v+1}$ proposes as soon as it receives $L_v$'s proposal, giving our protocols an \mvbp of $\delta$. Furthermore, since our pipelined protocols require two consecutive views to produce certified blocks before a new block can be committed, requirements (i) and (ii) also give our protocols a \mcl of $3\delta$.

\begin{figure}
    \includegraphics[width=0.48\textwidth]{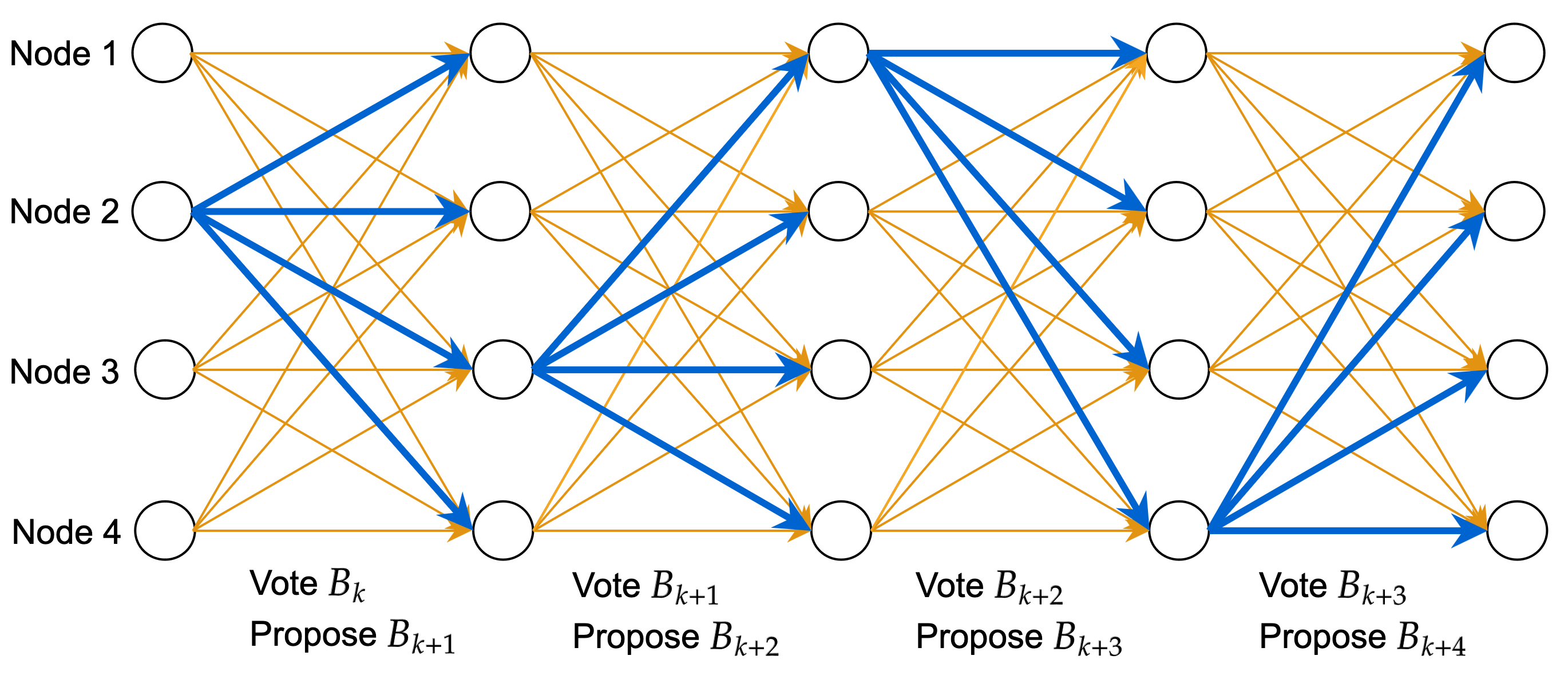}
    \caption{Optimistic proposal (pictured in blue) and vote multicasting (pictured in orange) enable \simple and \pipelined to propose new blocks at the same rate that they become certified when proposals and votes take equal time to propagate and process.}
    \label{fig:pipelined-moonshot-communication}
\end{figure}

\subsection{Protocol Details}
We define \simple in~\Cref{fig:simple-moonshot} as a series of event handlers to be run by each node $\node{i} \in \network$. We elaborate below.

\paragraph{Advance View and Timeout.} $\node{i}$ enters view $v$ from some view $v' < v$ upon receiving a view $v-1$ block certificate or a view $v-1$ timeout certificate (i.e. $\node{i}$ never decreases its local view). Before doing so, it first multicasts this certificate. This ensures that if the first honest node enters $v$ after GST then all honest nodes will enter $v$ or higher within $\Delta$ thereafter, helping our protocol to obtain liveness and reorg resilience. Subsequently, $\node{i}$ updates $\Lock_i$ to the highest ranked block certificate that it has received so far and if $\Lock_i$ is not $\Cert_{v-1}$ then $\node{i}$ unicasts a $\Status$ message containing $\Lock_i$ to $L_v$. We note that $\node{i}$ only updates $\Lock_i$ during the view transition process and does not do so after entering the new view, even if it receives a higher ranked block certificate. This ensures that the block certificate reported in a $\Status$ message corresponds to its honest sender's $\Lock_i$ for the duration of $v$, meaning that if $L_v$ waits to receive status messages from all honest nodes before proposing, then it is guaranteed to extend the block certified by the highest ranked block certificate locked by any honest node. Finally, $\node{i}$ enters $v$, resets $\ViewTimer_i$ to $5\Delta$ and starts counting down. If $L_v$ is honest and the network is synchronous then $\node{i}$ should enter $v+1$ within $5\Delta$ of entering $v$. If it does not, then it considers the current leader to have failed and so multicasts $\sig{\Timeout, v}_{i}$ to request a view change and prevent the protocol from halting. $\node{i}$ also does this whenever it observes that at least one other honest node has requested a view change for $v$.

\paragraph{Propose.} \simple allows two proposals to be created during view $v$: i) an optimistic proposal for view $v+1$, and; ii) a normal proposal for $v$. In the former case, $L_{v+1}$ multicasts $\sig{\OptPropose, B_{k+1}, v+1}$, where $B_{k+1}$ extends $B_k$, upon voting for $B_k$ in $v$, hoping that $B_k$ will become certified. When the protocol is operating in its happy path after GST, $B_k$ will indeed become certified, enabling voting for consecutive honest proposals to proceed without delay. In the latter case, $L_v$ multicasts $\sig{\Propose, B_h, \Cert_{v'}(B_{h-1}), v}$, where $B_h$ extends $B_{h-1}$, either upon receiving $\Cert_{v-1}(B_{h-1})$ within $2\Delta$ time of entering $v$, or after having $\Cert_{v'}(B_{h-1})$ as its highest block certificate after waiting for $2\Delta$ after entering $v$. Since messages are delivered within $\Delta$ time after GST, this $2\Delta$ wait ensures that $L_v$ will extend the highest certified block locked by any honest node when it proposes after GST, assisting with liveness and reorg resilience. We require $L_v$ to multicast a normal proposal even when it has already multicasted an optimistic proposal to ensure that it always produces a certified block when it proposes after GST. We note that this requirement can be removed from each of our protocols to obtain the corresponding leader-speaks-once variant, but doing so naturally sacrifices reorg resilience because the adversary can cause optimistic proposals to fail, even after GST. We discuss how after introducing the remaining protocol rules.

\paragraph{Vote.} \simple has two rules for voting, at most one of which each node may invoke at most once per view. Firstly, $\node{i}$ may $\Vote$ for an optimistic proposal containing $B_k$ proposed for view $v$ and extending $B_{k-1}$, when locked on $\Cert_{v-1}(B_{k-1})$. In the best case, $\node{i}$ receives the optimistic proposal containing $B_k$ and $\Cert_{v-1}(B_{k-1})$ simultaneously and so votes for $B_k$ immediately upon entering view $v$. Alternatively, if $\node{i}$ receives $\sig{\Propose, B_h, \Cert_{v'}(B_{h-1}), v}$ and $\Cert_{v'}(B_{h-1})$ ranks higher than or equal to $\Lock_i$ and $B_h$ extends $B_{h-1}$, then it votes for $B_h$. Importantly, if $L_v$ creates both an optimistic proposal and a normal proposal with the same parent, then since payloads are fixed for a given view, both proposals will contain the same block. This ensures that all honest nodes will vote for the same block, even if they use different vote rules.

\paragraph{Commit.} Finally, at any time during protocol execution, when an honest node $\node{i}$ receives $\Cert_{v-1}(B_{k-1})$ and $\Cert_{v}(B_{k})$, it commits $B_{k-1}$ and all of its uncommitted ancestors. We say that a node \textit{directly commits} $B_k$ and \textit{indirectly commits} any ancestors that it commits as a result of committing $B_k$.

\subsection{Analysis}
We now provide some discussion on the properties of \simple, including brief intuitions for its safety, liveness and reorg resilience. We give rigorous proofs for each of these properties in Appendix~\ref{sec:simple-proof}.

\paragraph{How does our protocol achieve safety?} The vote and commit rules together ensure that \simple satisfies the safety property of SMR. Specifically, if an honest node commits $B_k$ for view $v$ after receiving $\Cert_v(B_k)$ and $\Cert_{v+1}(B_{k+1})$, then a majority of the honest nodes must have voted for $B_{k+1}$ in view $v+1$. Therefore, since honest nodes vote at most once per view, an equivocating $\Cert_v$ cannot exist so no honest node will be able to commit any block other than $B_k$ for view $v$. Moreover, since the set of honest nodes that voted for $B_{k+1}$, say $H$, must have had $\Cert_v(B_k)$ when they voted for $B_{k+1}$, they will either lock this certificate or one of a higher rank upon transitioning from $v+1$ to a higher view. Once again then, since block certificates must contain votes from a majority of the honest nodes, every block certificate for every view greater than $v$ must contain a vote from at least one member of $H$. Suppose that $v'$ is the first view greater than $v+1$ to produce a block certificate and let $\node{i}$ be a member of $H$ that votes towards $\Cert_{v'}(B_l)$. Importantly, since $\node{i}$ must lock $\Cert_v(B_k)$ before voting for $B_l$ and since no higher ranked block certificate than $\Cert_{v+1}(B_{k+1})$ can exist before it does so, by the vote rules, $B_l$ must directly extend either $B_k$ or $B_{k+1}$. By extension then, every block certified for a higher view than $v$ must extend $B_k$. This is sufficient to ensure safety.

\paragraph{Why propose twice?} As previously mentioned, we require our leaders to make normal proposals even if they have already made an optimistic proposal because the adversary can cause optimistic proposals to fail even after GST. Suppose that an honest leader $L_v$ proposes $B_k$ extending $B_{k-1}$ in an optimistic proposal. Per the optimistic vote rule, the adversary can cause $B_k$ to fail by preventing some honest node from locking $\Cert_{v-1}(B_{k-1})$. This could happen either due to some other block, say $B_l$, becoming certified for view $v-1$, or due to the node entering $v$ via $\TimeoutCert_{v-1}$. In either case, since $L_v$ is guaranteed to observe the highest ranked block certificate locked by any honest node upon entering view $v$, say $\Cert_{v'}(B_h)$, before it multicasts its normal proposal, it will be able to multicast a new block, say $B_{h+1}$, that extends $B_h$. Therefore, since honest nodes only update their locks when entering a new view, those that receive $B_{h+1}$ whilst in view $v$ will all have $\Lock_i \le \Cert_{v'}(B_h)$ and hence will vote for it. Thus, this requirement yields two important properties: i) that honest leaders are able to correct themselves when they initially extend a block that fails to become certified, and; ii) that if this block does become certified and its certificate is locked by any honest node, then the block included in the optimistic proposal will become certified even if some honest nodes initially fail to lock this certificate. This ensures that every honest leader that proposes after GST produces exactly one certified block.

\paragraph{How does our protocol achieve reorg resilience and liveness?}
As we have just explained, \simple guarantees that every honest leader that proposes after GST produces exactly one certified block. Suppose that $L_v$ is such an honest leader and produces $\Cert_v(B_k)$, and let $t$ denote the time that the first honest node enters $v$. Since the multicasting of block certificates and timeout certificates ensures that all honest nodes will enter view $v$ or higher within $t + \Delta$, if $L_v$ is honest then it will send its last proposal by $t + 3\Delta$, so all honest nodes will finish voting before $t + 4\Delta$ and thus before any honest node can have sent $\TimeoutMessage_v$ or higher. Therefore, either all honest nodes vote for $B_k$ or some honest node must have entered $v+1$ via $\Cert_v(B_k)$ first. In either case, all honest nodes will receive $\Cert_v(B_k)$ within $5\Delta$ of the first honest node entering $v$, so at least $f+1$ honest nodes will lock this certificate. Therefore, since these $f+1$ nodes will not vote for any optimistic proposal that does not directly extend their lock, and since no higher ranked block certificate can exist before they lock $\Cert_v(B_k)$, every certified block for every view greater than $v$ must extend $B_k$ satisfying \Cref{def:reorg_resilience}. Moreover, when $L_{v+1}$ is also honest, it will necessarily be among the $f+1$ honest nodes that lock $\Cert_v(B_k)$ and will therefore multicast a proposal that extends $B_k$ no later than the time that it enters $v+1$. Consequently, by the prior reasoning, all honest nodes will also receive $\Cert_{v+1}(B_{k+1})$ and thus will commit $B_k$. Accordingly, \simple commits a new block whenever there are two consecutive honest leaders after GST, which is sufficient to ensure liveness. 

\paragraph{Communication complexity.}
Per~\Cref{fig:simple-moonshot}, \simple requires nodes to multicast votes and timeout messages. Each of these messages are $O(1)$ words in size, giving these actions a network-level communication complexity of $O(n^2)$ words per view. Optimistic proposals are likewise $O(1)$ in size, as are normal proposals when threshold signatures are used to compress the vote signatures used to construct block certificates. Accordingly, since the proposal actions are only performed by the leader of the view, they incur $O(n)$ words per view. The forwarding of $\Lock_i$ to the next leader upon advancing to a new view incurs a similar cost, assuming threshold signatures, while the multicasting of certificates incurs $O(n^2)$ communication. Overall then, \simple exhibits a network-level communication complexity of $O(n^2)$ words per view, assuming threshold signatures.

\section{\pipelined}\label{sec:moonshot_v2}
Although \simple has \mvbp $= \delta$, \mcl $= 3\delta$ and reorg resilience, it only provides responsiveness under consecutive honest leaders. If the leader of the current view fails then the next leader has to wait for $\Omega(\Delta)$ time to ensure that it can create a block that will become certified, naturally increasing \vl. We now present \pipelined (\Cref{fig:pipelined-moonshot}), a CRL protocol that improves on \simple in both of these areas to achieve full optimistic responsiveness and a \vl of $3\Delta$.

\begin{figure*}[!ht]
\small
    \begin{boxedminipage}[t]{\textwidth}
    A \pipelined node $\node{i}$ runs the following protocol whilst in view $v$:
    \begin{enumerate}[leftmargin=*] 
        \item \textbf{Propose.} Upon entering $v$ and after executing \textit{Advance View} and \textit{Lock}, if $\node{i}$ is $L_v$, propose using one of the following rules:
        \begin{enumerate}
            \item \textbf{Normal Propose.} If $L_{v}$ entered $v$ by receiving $\Cert_{v-1}(B_{k-1})$, multicast $\sig{\Propose, B_{k}, \Cert_{v-1}(B_{k-1}), v}$ such that $B_k$ extends $B_{k-1}$.

            \item \textbf{Fallback Propose.} If $L_{v}$ entered $v$ by receiving $\TimeoutCert_{v-1}$, multicast $\sig{\FallbackPropose, B_k, \Cert_{v'}(B_{k-1}), \TimeoutCert_{v-1}, v}$ such that $\Cert_{v'}(B_{k-1})$ is $\Lock_i$ and $B_k$ extends $B_{k-1}$.
        \end{enumerate}

        \item \textbf{Vote.}
        $\node{i}$ votes at most twice in view $v$ when the following conditions are met:
        \begin{enumerate}
            \item \textbf{Optimistic Vote.} Upon receiving $\sig{\OptPropose, B_{k}, v}$ such that $B_k$ extends $B_{k-1}$, if (i) $\TimeoutView_i < v-1$, (ii) $\Lock_{i} = \Cert_{v-1}(B_{k-1})$ and (iii) $\node{i}$ has not voted in $v$, multicast $\sig{\OptVote, H(B_k), v}_i$.

            \item After executing \textit{Advance View} and \textit{Lock} with all embedded certificates, vote once when one of the following conditions are satisfied:
            \begin{enumerate}
                \item  \textbf{Normal Vote.} Upon receiving $\sig{\Propose, B_{k}, \Cert_{v-1}(B_h), v}$, if (i) $\TimeoutView_i < v$, (ii) $B_k$ directly extends $B_h$ and (iii) $\node{i}$ has not sent an optimistic vote for an equivocating block $B'_{k'}$ in $v$, multicast $\sig{\Vote, H(B_k), v}_i$.

                \item \textbf{Fallback Vote.} Upon receiving $\sig{\FallbackPropose, B_{k}, \Cert_{v'}(B_h), \TimeoutCert_{v-1},v}$ if (i) $\TimeoutView_i < v$, (ii) $B_k$ directly extends $B_h$ and (iii) $\Cert_{v'}(B_h)$ has an equal or greater rank than the highest ranked certificate in $\TimeoutCert_{v-1}$, multicast $\sig{\FallbackVote, H(B_k), v}_i$.
            \end{enumerate}
        \end{enumerate}

        \item \textbf{Optimistic Propose.} Upon voting for $B_{k}$ in view $v$, if $\node{i}$ is $L_{v+1}$, multicast $\sig{\OptPropose, B_{k+1}, v+1}$ such that $B_{k+1}$ extends $B_{k}$.

        \item \textbf{Timeout.} Upon the expiration of $\ViewTimer_i$, if $\node{i}$ has not already sent $\TimeoutMessage_v$, multicast $\sig{\Timeout, v, \Lock_i}_{i}$ and set $\TimeoutView_i = \max(\TimeoutView_i, v)$. Additionally, upon receiving $f+1$ distinct $\sig{\Timeout, v', \_}_*$ messages or $\TimeoutCert_{v'}$ such that $v' \ge v$ and not having sent $\TimeoutMessage_{v'}$, multicast $\sig{\Timeout, v', \Lock_i}_{i}$ and set $\TimeoutView_i = \max(\TimeoutView_i, v')$.

        \item \textbf{Advance View.} $\node{i}$ enters $v'$ where $v' > v$ using one of the following rules:
        \begin{itemize}
            \item[-] Upon receiving $\Cert_{v'-1}(B_h)$. Also, multicast $\Cert_{v'-1}(B_h)$.
            \item[-] Upon receiving $\TimeoutCert_{v'-1}$. Also, unicast $\TimeoutCert_{v'-1}$ to $L_{v'}$.
        \end{itemize}
        Finally, reset $\ViewTimer_i$ to $3\Delta$ and start counting down.
    \end{enumerate}

    $\node{i}$ additionally performs the following actions in any view:
    \begin{enumerate}[leftmargin=*]
        \item \textbf{Lock.}\label{step:lock2} Upon receiving $\Cert_v(B_k)$ in any protocol message whilst having $\Lock_i = \Cert_{v'}(B_{k'})$ such that $v > v'$, set $\Lock_i$ to $\Cert_v(B_k)$.

        \item \textbf{Direct Commit.}\label{step:direct-commit-2} Upon receiving $\Cert_{v-1}(B_{k-1})$ and $\Cert_{v}(B_{k})$ such that $B_{k}$ extends $B_{k-1}$, commit $B_{k-1}$.
        
        \item \textbf{Indirect Commit.}\label{step:indirect-commit-2} Upon directly committing $B_{k-1}$, commit all of its uncommitted ancestors.
    \end{enumerate}
\end{boxedminipage}
\caption{The \pipelined Protocol}
\label{fig:pipelined-moonshot}
\end{figure*}

\paragraph{Towards achieving optimistic responsiveness with \vl $=3\Delta$.}\label{sec:pipelined_reasoning}
In \pipelined, we separate the fallback case of \simple's normal proposal into its own proposal type by enabling $L_v$ to create a \emph{fallback proposal} extending its $\Lock$ upon entering $v$ via $\TimeoutCert_{v-1}$. While this means that $L_v$ no longer needs to wait $\Omega(\Delta)$ time before proposing in the fallback path---making \pipelined optimistically responsive per~\Cref{def:responsiveness}---it also means that $L_v$ may not receive the locks of all honest nodes before proposing. However, since $\TimeoutCert$s must be constructed from $2f+1$ timeout messages, which in turn must now include the sender's $\Lock$, $L_v$ must process the locks of at least $f+1$ honest nodes before creating its proposal. Consequently, $L_v$'s $\Lock$ is guaranteed to have a rank at least as great as the highest the highest ranked $\Lock$ among these nodes at the time that they sent their timeout messages. This, along with the rules for voting, guarantees that there cannot exist a committable block with a higher height than $L_v$'s $\Lock$. This helps to preserve safety in light of the additional modification that we make to preserve liveness, which we do by allowing $\node{i}$ to vote for $\sig{\FallbackPropose, B_{k}, \Cert_{v'}(B_h), \TimeoutCert_{v-1},v}$ even if it has $\Lock_i > \Cert_{v'}(B_h)$, given $B_{k}$ directly extends $B_h$ and $\Cert_{v'}(B_h)$ has a rank at least as great as the highest ranked block certificate included in $\TimeoutCert_{v-1}$.

Requiring timeout messages to include block certificates naturally increases their size. Similarly, since $\TimeoutCert$s must provably contain the highest ranked block certificate out of $2f+1$ timeout messages, they are necessarily linear in size even when using threshold signatures~\cite{boneh2004short}. Accordingly, to avoid cubic communication complexity even under threshold signatures, our protocol replaces the $\TimeoutCert$ multicast of \simple with a Bracha-style amplification step~\cite{bracha1987asynchronous}. In particular, $\node{i}$ multicasts a $\TimeoutMessage_v$ whilst in view $v'$ where $v' \le v$ when it first receives either $f+1$ $\TimeoutMessage_v$ or $\TimeoutCert_v$ from other nodes. This ensures that all honest nodes continue to enter new views after GST: In short, either all honest nodes will send view $v$ Timeout messages, or, since we still require nodes to multicast block certificates, either some honest node must have observed and multicasted a view $v$ or higher block certificate, or all honest nodes will send view $v''$ Timeout messages, where $v'' > v$.

\paragraph{Linear timeout certificates without threshold signatures.} Block certificates are necessarily linear in size without the use of threshold signatures. Consequently, to avoid $O(n^2)$-sized timeout certificates in this setting, a node may sign only the view number of the block certificate included in its timeout message instead of the full block certificate. This allows $\TimeoutCert_v$ to be constructed from $2f+1$ such signatures mapped to their corresponding block certificate view numbers, and the full highest-ranked block certificate. We observe that this block certificate must be included for the timeout certificate to be able to guarantee the existence of a block certificate for the highest reported view number.

\subsection{Protocol Details}
We now present the details of \pipelined. We start with refinements to the definition of a block certificate and the certificate ranking rules before elaborating on the steps outlined in~\Cref{fig:pipelined-moonshot} that differ from \simple.

\paragraph{Block certificates.} In \pipelined, we use three types of signed vote messages: an optimistic vote ($\OptVote$), a normal vote ($\Vote$) and a fallback vote ($\FallbackVote$). Importantly, vote messages with different types may not be aggregated together. Accordingly, we now distinguish between three different types of block certificates. An \emph{optimistic certificate} $\Cert^{o}_v(B_h)$ for a block $B_h$ consists of $2f+1$ distinct $\OptVote$ messages for $B_h$ for view $v$. Similarly, a \emph{normal certificate} $\Cert^{n}_v(B_h)$ consists of $2f+1$ distinct $\Vote$ messages for $B_h$ for view $v$. Finally, a \emph{fallback certificate} $\Cert^{f}_v(B_h)$ consists of $2f+1$ distinct $\FallbackVote$ messages for $B_h$ for view $v$. We denote a block certificate with $\Cert_v(B_h)$ whenever its type is not relevant.

\paragraph{Locking.} \simple only allowed $\node{i}$ to update $\Lock_i$ upon entering a new view. In contrast, \pipelined requires $\node{i}$ to update $\Lock_i$ upon receiving a higher ranked block certificate than its current $\Lock_i$, which may happen at any time during the protocol run.

\paragraph{Advance View and Timeout.} As in \simple, $\node{i}$ enters view $v$ from some view $v' < v$ upon receiving $\Cert_{v-1}$ or $\TimeoutCert_{v-1}$. In the former case, as before, it then multicasts $\Cert_{v-1}$ to assist with reorg resilience and view synchronization. Comparatively, in the latter case $\node{i}$ now unicasts $\TimeoutCert_{v-1}$ to $L_v$ instead of multicasting it. This helps to reduce the communication complexity of the protocol in light of its modified timeout messages, while still ensuring that $L_v$ enters $v$ within $\Delta$ of the first honest node doing so after GST. This in turn makes a $\ViewTimer$ of $3\Delta$ sufficient to guarantee the liveness of the protocol (which can be further optimized under crashed leaders, as explained in Appendix~\ref{sec:view_length_optimiztion}), which $\node{i}$ additionally resets regardless of how it enters $v$, and starts counting down. As before, if $\node{i}$ does not advance to a new view before its view timer expires then it multicasts $\sig{\Timeout, v, \Lock_i}_{i}$. It likewise multicasts the same message for $v''$ upon observing evidence of at least one honest node requesting a view change for $v''$ such that $v'' \ge v$. This latter rule differs from \simple and compensates for \pipelined's removal of $\TimeoutCert$ multicasting.

\paragraph{Propose.} \pipelined consists of three distinct ways to propose a new block in a view; i) an optimistic proposal, ii) a normal proposal, and iii) a fallback proposal. An honest node proposes using at most two of the three methods. The optimistic proposal rule remains the same as in \simple and serves the same purpose, allowing voting to proceed without delay when network conditions are favourable. Comparatively, the normal proposal rule now only captures the first case of the same rule in \simple: Namely, $L_v$ multicasts a normal proposal $\sig{\Propose, B_k, \Cert_{v-1}(B_{k-1}), v}$, where $B_k$ extends $B_{k-1}$, upon entering view $v$ via $\Cert_{v-1}(B_{k-1})$. As before, $L_v$ does this even if it has already sent an optimistic proposal extending $B_{k-1}$ (which, as before, will necessarily contain $B_k$). As in \simple, this helps \pipelined obtain reorg resilience by ensuring that, after GST, $L_v$ will create a proposal that all honest nodes will vote for. Finally, $L_v$ multicasts $\sig{\FallbackPropose, B_h, \Cert_{v'}(B_{h-1}), \TimeoutCert_{v-1}, v}$, where $B_h$ extends $B_{h-1}$ and $\Cert_{v'}(B_{h-1})$ is $\Lock_i$, upon entering $v$ via $\TimeoutCert_{v-1}$. Importantly, since $L_v$ only attempts this proposal after executing the \emph{Lock} rule, $\Cert_{v'}(B_{h-1})$ is guaranteed to have a rank greater than or equal to that of the highest ranked certificate included in $\TimeoutCert_{v-1}$. We note that in the case of fallback proposals, the requirement that blocks created for the same view must contain the same payload can be relaxed because the voting rules ensure that if $\Cert^{f}_v(B_h)$ exists then no other block can be certified for $v$.

\paragraph{Vote.} In \pipelined, $\node{i}$ may vote up to twice in a view; at most once for an optimistic proposal and at most once for either a normal proposal or a fallback proposal. More precisely, $\node{i}$ multicasts $\sig{\OptVote, H(B_k), v}_i$ for $\sig{\OptPropose, B_k, v}$, where $B_k$ extends $B_{k-1}$, when in view $v$ if it has not yet sent a vote for $v$, or a timeout message for $v-1$ or higher, and has locked $\Cert_{v-1}(B_{k-1})$. As before, this enables $\node{i}$ to vote for $B_k$ immediately upon entering $v$ in the best case. Additionally, $\node{i}$ sends $\sig{\Vote, H(B_k), v}_i$ for $\sig{\Propose, B_{k}, \Cert_{v-1}(B_{k-1}), v}$ when in $v$ if it has not sent either an $\OptVote$ for an equivocating block in view $v$ or a timeout message for view $v$ or higher, and $B_k$ extends $B_{k-1}$. Importantly, $\node{i}$ must send this vote if it has already sent an optimistic vote for $B_k$. This ensures that $B_k$ will be certified when $L_v$ is honest and proposes after GST in the case where some honest nodes are unable to send an optimistic vote for $B_k$. Otherwise, $\node{i}$ multicasts $\sig{\FallbackVote, H(B_h), v}_i$ for $\sig{\FallbackPropose, B_h, \Cert_{v'}(B_{h-1}), \TimeoutCert_{v-1},v}$ when in view $v$ if it has not sent a timeout message for view $v$ or higher, $B_h$ extends $B_{h-1}$ and $\Cert_{v'}(B_{h-1})$ has a rank greater than or equal to that of the highest ranked block certificate in $\TimeoutCert_{v-1}$. Notice that this rule allows $\node{i}$ to send a fallback vote for $B_h$ after having sent an optimistic vote for an equivocating block, say $B_k$. However, since the fallback proposal containing $B_h$ can only be valid if it contains $\TimeoutCert_{v-1}$, at least $f+1$ honest nodes must have sent $\TimeoutMessage_{v-1}$ before entering view $v$ and thus will not be able to trigger the optimistic vote rule for $B_k$, so $\Cert^o_v(B_k)$ will never exist.

\subsection{Analysis}
We now briefly analyze the safety and communication complexity of \pipelined. Detailed proofs of \pipelined's safety, liveness and reorg resilience can be found in Appendix~\ref{sec:pipelined-proof}.

\paragraph{Why is it safe to vote for a fallback proposal?} As we mentioned earlier, we require honest nodes to vote for valid fallback proposals even when they are locked on a higher ranked block certificate than that of the parent of the proposed block. This remains safe because a fallback proposal must be justified by a $\TimeoutCert$ for the previous view, which in turn contains information about the locks of a majority of the honest nodes. Specifically, $\TimeoutCert_v$ guarantees that at least $f+1$ nodes had yet to vote for a higher height than $h+1$ upon sending $\TimeoutMessage_v$, where $h$ is the height of $\Cert_{v'}(B_h)$, the highest ranked block certificate included in $\TimeoutCert_v$. Consequently, there cannot exist a committable block for any height greater than $h$ when $\TimeoutCert_v$ is constructed. Moreover, if any block can be committed at height $h$ then there can be only one such block. This is because the commit rule only allows a block at height $h$ proposed for view $v''$ to be committed if its child becomes certified in $v''+1$. Therefore, if $B_h$ can be committed then at least $f+1$ honest nodes must have voted for its child in $v'+1$, and since an honest node cannot vote for a block unless it possesses the block certificate for its parent, these nodes must have had $\Cert_{v'}(B_h)$ when they did so. Consequently, every $\TimeoutCert$ for $v'+1$ or higher will necessarily contain $\Cert_{v'}(B_h)$ or a block certificate for one of its descendants as its highest ranked block certificate, meaning that every fallback proposal for $v'+1$ or higher will necessarily extend $B_h$. Moreover, by extension, so will every subsequent optimistic or normal proposal.

\paragraph{Communication complexity.}
As previously mentioned, \pipelined requires nodes to include $\Lock_i$ in their timeout messages to ensure that $\TimeoutCert$s attest to the highest lock amongst at least $f+1$ honest nodes. This naturally makes timeout certificates and fallback proposals at least $O(n)$ words in size. Comparatively, as in \simple, optimistic proposals are $O(1)$ in size, as are normal proposals when threshold signatures are used, giving the proposal action a network-level communication complexity of $O(n^2)$ words per view. Similarly, the multicasting of $O(1)$ sized timeout messages (when threshold signatures are used), $\Vote$ messages and block certificates by all nodes, and the forwarding of timeout certificates by all nodes to the next leader also incur $O(n^2)$ communication. Accordingly, \pipelined has a network-level communication complexity of $O(n^2)$ words per view when threshold signatures are used.

\section{\commit}\label{sec:commit-moonshot}
Until now, we have measured \mcl in terms of $\delta$. However, this is imprecise because $\delta$ provides no way of differentiating between the performance of protocols that exchange one type of message for another. The pipelining technique fundamentally facilitates the removal of one or more phases from a protocol by granting another phase additional meaning. In existing pipelined consensus protocols, this technique replaces two (or more) consecutive phases of voting for one block proposal, with one phase of voting for two (or more) consecutive block proposals. This means that the commit latency of a block in the pipelined setting is proportional to the dissemination time of not only the block itself, but also its child (in the best case). More to the point, pipelining essentially exchanges the cost of disseminating additional votes for the cost of disseminating additional proposals and thus naturally increases commit latency when proposals take sufficiently longer to disseminate than votes.

\begin{figure}[!t]
\small
    \begin{boxedminipage}[ht]{\linewidth}
    \commit can be obtained by adding the following rules to the protocol for $\node{i}$ presented in~\Cref{fig:pipelined-moonshot}:
        
    \begin{enumerate}[leftmargin=*]        
        \item \textbf{Direct Pre-commit.} Upon receiving $\Cert_v(B_k)$ whilst in any view $v'$ such that $v' \le v$, if $\TimeoutView_i < v$, multicast $\sig{\Commit, H(B_k), v}_i$.

        \item \textbf{Indirect Pre-commit.} Upon receiving $\Cert_v(B_k)$ whilst in any view, having multicasted a commit vote for any descendant of $B_k$, having $\TimeoutView_i < v$ and having not yet multicasted $\sig{\Commit, H(B_k), v}_i$, multicast $\sig{\Commit, H(B_k), v}_i$.

        \item \textbf{Alternative Direct Commit.}\label{step:direct-commit-3} Upon receiving a quorum of distinct $\sig{\Commit, H(B_k), v}_*$ whilst in any view, commit $B_k$.
    \end{enumerate}
    \end{boxedminipage}
    \caption{\commit}
    \label{fig:commit_moonshot}
\end{figure}

We characterize this behavior using a communication model based on the \emph{modified partially synchronous model} of Blum et al.~\cite{blum2023analyzing} in which we assume that small messages (in this case, votes) are delivered within $\rho$ time while large messages (in this case, block proposals) are delivered within $\beta$ time. Moreover, we assume that after GST $\rho = [0, \min(\beta))$ and $\beta = (\max(\rho), \Delta]$. Under this model, \simple and \pipelined both incur \mcl $ = 2\beta + \rho$.

\begin{figure}
    \includegraphics[width=0.48\textwidth]{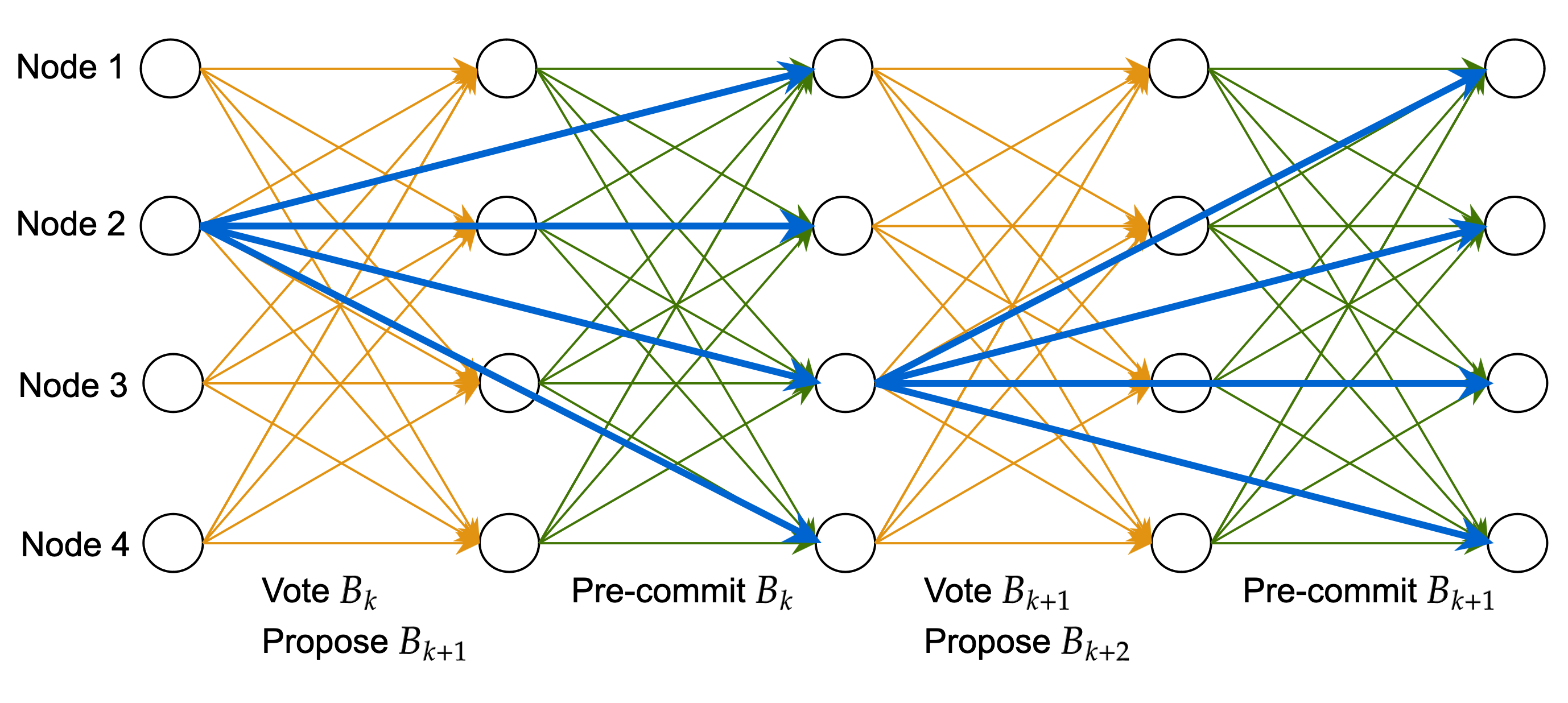}
    \caption{Explicit commit votes (pictured in green) enable \commit to commit blocks sooner than its pipelined counterparts when block proposals (pictured in blue) take sufficiently longer to disseminate than votes.}
    \label{fig:commit-moonshot-advantage}
\end{figure}

We now present a protocol with \mcl $= \beta + 2\rho$, which we call \commit. Accordingly, when $\rho < \beta$ (as in~\Cref{fig:commit-moonshot-advantage}), which we assume is typical in practice, this protocol provides improved commit latency over those previously presented by integrating an explicit \emph{pre-commit} phase. Like its counterparts, \commit also obtains \mvbp $= \delta$ ($\beta$) and provides both reorg resilience and optimistic responsiveness. Additionally, while \simple and \pipelined require two consecutive honest leaders to guarantee a commit after GST, \commit requires only one.

We present the modifications required to convert \pipelined to \commit in~\Cref{fig:commit_moonshot}. Since \commit retains the rules of \pipelined, the same liveness argument that can be made for the latter also applies to the former. However, the introduction of a secondary commit path demands additional reasoning about the safety of the protocol. We present a brief intuition to this end below and provide complete proofs in Appendix~\ref{sec:cm_proof}.

\paragraph{Safety intuition.} Per the alternative commit rule given in~\Cref{fig:commit_moonshot}, $\node{i}$ commits $B_k$ and all of its uncommitted ancestors upon receiving a quorum (i.e. $2f+1$ when $n = 3f+1$) of distinct $\sig{\Commit, H(B_k), v}_*$ messages. This remains safe because $2f+1$ such messages can only exist if at least $f+1$ honest nodes do not send $\TimeoutMessage_v$. Consequently, if any block becomes certified for $v+1$ then it must have been proposed in either an optimistic or normal proposal and thus must be a child of $B_k$. Otherwise, $\TimeoutCert_{v+1}$ will contain $\Cert_v(B_k)$ as its highest ranked block certificate and therefore every subsequently certified block will necessarily extend $B_k$.

\section{Implementation and Evaluation}\label{sec:eval}
As shown in~\Cref{tbl:related_work}, \pipelined and \commit equal or surpass the theoretical performance of prior $O(n^2)$ CRL protocols in all considered metrics. The primary question that remains, then, is whether their increased communication complexity relative to linear protocols is justified. Accordingly, we decided to implement our protocols and evaluate them against Jolteon, a linear protocol with state-of-the-art performance in most metrics and several high-quality open-source implementations.

\paragraph{Implementation.} We implemented all three of our protocols by modifying the code for Jolteon available in the Narwhal-HotStuff branch of the repository~\cite{narwhal-hs} created by Facebook Research for evaluating Narhwal and Tusk~\cite{danezis2022narwhal}. We decoupled our implementation from Narhwal and did the same for Jolteon so that we could compare the two consensus protocols in isolation. We replaced both the Narwhal mempool and the simulated-client process by having the leaders of each protocol create parametrically sized payloads during the block creation process, with individual payload items being 180 bytes in size. We used ED25519 signatures and constructed certificate proofs from an array of these signatures. We left the TCP-based network stack mostly intact and applied the few necessary changes to both implementations to ensure that any differences in performance were solely due to the differences between the consensus protocols themselves.

\begin{table}[t]
\footnotesize
    \caption{Observed Latencies (in ms) between AWS regions\label{tab:observed_latency}}
\begin{center}
    \begin{tabular}{| c | c | c | c | c | c | }
        \hline
         & \multicolumn{5}{|c|}{\textbf{Destination}$^*$} \\
        \hline
        \textbf{Source} & us-e-1 & us-w-1  & eu-n-1 & ap-ne-1 & ap-se-2 \\
        \hline
        us-east-1  & 5.23 & 61.87 & 113.78 & 167.6 & 197.42 \\ 
        us-west-1  & 62.88 & 3.69 & 172.17 & 109.89 & 141.54 \\ 
        eu-north-1 & 114.09 & 173.31 & 5.48 & 248.67 & 271.68 \\ 
        ap-northeast-1 & 168.04 & 109.94 & 251.63 & 5.99 & 111.67 \\ 
        ap-southeast-2 & 199.54 & 146.06 & 272.31 & 112.11 & 4.53 \\ 
        \hline
        \multicolumn{6}{c}{$^*$Region names are abbreviated versions of the Source regions.}
    \end{tabular}
\end{center}
\end{table}

\paragraph{Setting.} We chose to perform our evaluation in a setting typical of modern low-latency public blockchains such as Aptos~\cite{aptos} to demonstrate the efficacy of our protocols when network latency is the dominating performance factor. Accordingly, we constructed moderately-sized (up to 200 nodes) wide-area networks of nodes with high bandwidth capabilities and moderate computational capabilities. Specifically, we distributed the nodes evenly across the us-east-1 (N. Virginia), us-west-1 (N. California), eu-north-1 (Stockholm), ap-northeast-1 (Tokyo) and ap-southeast-2 (Sydney) AWS regions, with each node being allocated its own m5.large EC2 instance and connected to every other node via a separate point-to-point link. Each instance ran Ubuntu 20.04 and had a network bandwidth of up to 10Gbps\footnote{https://docs.aws.amazon.com/AWSEC2/latest/UserGuide/ec2-instance-network-bandwidth.html}, 8GB of memory and Intel Xeon Platinum 8000 series processors with 2 virtual cores. \Cref{tab:observed_latency} reports the typical (90th percentile) latencies observed between these regions around the time of our experiments.

\paragraph{Variables and Metrics.} We first evaluated the trade-off between \mcl, \mvbp and steady-state communication complexity in this setting by running all protocols with $f'=0$, where $f'$ denotes the number of actual failures in the system (i.e. $f' \le f = \floor{\frac{n-1}{3}}$), under varying network and payload sizes. Subsequently, we evaluated the impact of \vl, reorg resilience, pipelining and optimistic responsiveness by running all protocols in a fixed network with $f' = f$ and varying leader schedules. We measured these trade-offs by comparing the throughput and latency of each protocol and established two metrics for throughput: Firstly, the number of blocks committed by at least $2f+1$ nodes during a run, hereafter referred to as \textit{throughput}; and secondly, the average number of bytes of payload data from (subsequently) committed blocks transferred per second (i.e., throughput $\times$ payload size $\div$ runtime), hereafter referred to as \textit{transfer rate}. For latency, we measured the average time between the creation of a block and its commit by the $(2f+1)\text{-th}$ node. The plotted results are the averages of the related metrics across three five minute runs for each related configuration of the system.

We refer to \simple, \pipelined, \commit and Jolteon as \SimpleShort, \PipelinedShort, \CommitShort and J in the accompanying figures and tables.

\begin{table}[t]
    \caption{Performance vs Jolteon ($f'=0$)\label{tab:performance_vs_jolteon}}
\begin{center}
    \begin{tabular}{| c | c | c | c | c | c | c | c | c |}
        \hline
         & \multicolumn{4}{|c|}{ \textbf{Throughput Increase ($\%$)} } & \multicolumn{4}{|c|}{ \textbf{Latency Reduction ($\%$)} }\\
        \hline
        \textbf{Prot.} & \textbf{Max} & \textbf{$\Bar{x}$} & \textbf{$\Tilde{x}$} & \textbf{Min} & \textbf{Max} & \textbf{$\Bar{x}$} & \textbf{$\Tilde{x}$} & \textbf{Min}\\
        \hline
        \SimpleShort & $230$ & $70$ & $55$ & $33$ & $72$ & $46$ & $43$ & $37$ \\
        \PipelinedShort & $230$ & $68$ & $55$ & $24$ & $72$ & $45$ & $42$ & $32$ \\
        \CommitShort & $214$ & $66$ & $55$ & $25$ & $71$ & $56$ & $61$ & $38$ \\
        \hline
    \end{tabular}
\end{center}
\end{table}

\begin{table}[t]
    \caption{Performance vs Jolteon ($f'=0$, Outliers Removed)\label{tab:performance_vs_jolteon_no_outliers}}
\begin{center}
    \begin{tabular}{| c | c | c | c | c | c | c | c | c |}
        \hline
         & \multicolumn{4}{|c|}{ \textbf{Throughput Increase ($\%$)} } & \multicolumn{4}{|c|}{ \textbf{Latency Reduction ($\%$)} }\\
        \hline
        \textbf{Prot.} & \textbf{Max} & \textbf{$\Bar{x}$} & \textbf{$\Tilde{x}$} & \textbf{Min} & \textbf{Max} & \textbf{$\Bar{x}$} & \textbf{$\Tilde{x}$} & \textbf{Min}\\
        \hline
        \SimpleShort    & $72$ & $53$ & $55$ & $33$ & $56$ & $43$ & $42$ & $37$ \\
        \PipelinedShort & $70$ & $51$ & $54$ & $24$ & $56$ & $43$ & $42$ & $32$ \\
        \CommitShort    & $74$ & $52$ & $54$ & $25$ & $69$ & $54$ & $58$ & $38$ \\
        \hline
    \end{tabular}
\end{center}    
\end{table}

\begin{figure}
    \centering
    \includegraphics[width=0.42\textwidth]{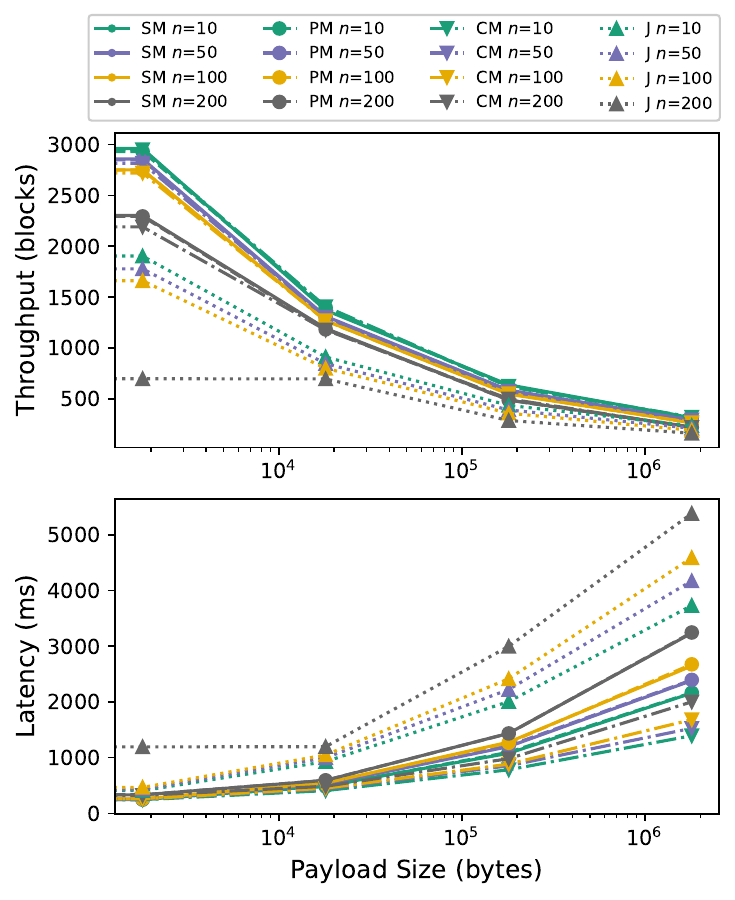}
    \caption{Performance Overview ($f'=0$, $p\le1.8$MB). Key trends: (1) Throughput approximately halves and latency roughly doubles for every order of magnitude increase in $p$. (2) Performance in both metrics decreases for all protocols as the network size increases. (3) Our protocols perform similarly in terms of throughput; Commit Moonshot achieves increasingly better latency as $p$ increases. (4) Our protocols outperform Jolteon in both metrics.}
    \label{fig:overview}
\end{figure}

\subsection{Happy Path Evaluation ($f'=0$)}
We initially tested networks of 10, 50, 100 and 200 honest nodes with block payload sizes ranging from empty to 1.8MB to understand how the tested protocols scale under an increasing communication load. \Cref{fig:overview} reports the results of these experiments. We subsequently tested additional payload sizes in the 200 node network to discover the approximate maximum transfer rate of each protocol in this setting, which can be seen in~\Cref{fig:throughput-v-latency}.

As shown in~\Cref{fig:overview}, \Cref{fig:improvement} and~\Cref{tab:performance_vs_jolteon_no_outliers}, all \moonshot protocols produced notably higher throughput than Jolteon in all tested configurations due to the more frequent block production afforded by their reduced \mvbp. Likewise, the reduction in \mcl achieved by multicasting votes (in conjunction with optimistic proposals, in the case of the pipelined protocols) also caused them to produce substantially decreased latency compared to Jolteon across all configurations. The 200 node network produced significant outliers under the empty and 1.8kB payload configurations, with all three protocols exhibiting about thrice the throughput and a quarter of the latency of Jolteon, compared to the approximately \emph{$50\%$ increase in throughput} and \emph{$40\%-50\%$ reduction in latency} seen on average across all other configurations. \simple and \pipelined produced near-identical performance in both metrics for most configurations due to the similarity of their happy-path protocols. Conversely, although \commit produced similar throughput to these protocols, it exhibited substantially reduced latency for payloads above 18kB due to its explicit commit messages, clearly showing the inefficiency of pipelining when blocks are large. Generally speaking, all three \moonshot protocols produced increasingly higher throughput and relatively consistent improvements to latency compared to Jolteon as the network size increased, showing that obtaining linear communication complexity is counter-productive in WANs of this scale if it comes at the cost of sequentializing network operations (i.e., reducing \mvbp and \mcl). Finally, per~\Cref{fig:throughput-v-latency}, all three \moonshot protocols achieved a \emph{higher maximum transfer rate with lower latency} than Jolteon in the 200 node network, with \commit producing the best results. Overall, these results show that the happy paths of our \moonshot protocols scale well and provide meaningfully decreased latency and increased throughput compared to Jolteon under the experimental conditions, with \commit being the most efficient option.

\begin{figure}
    \centering
    \includegraphics[width=0.42\textwidth]{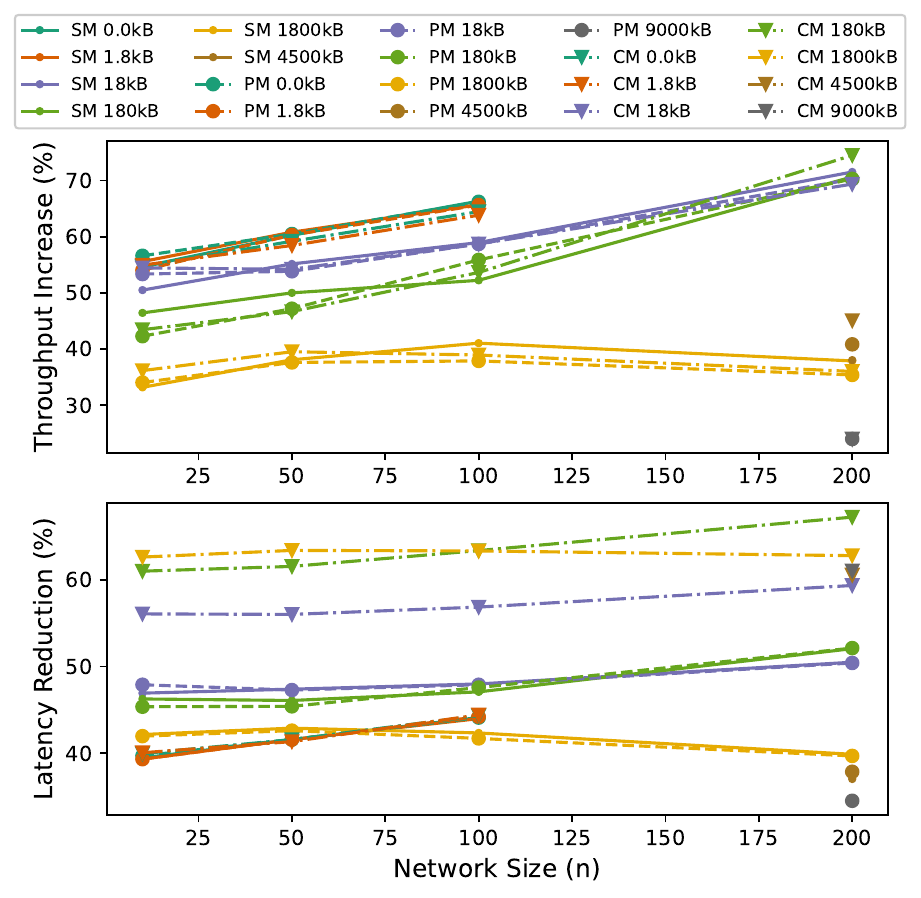}
    \caption{Performance vs. Jolteon ($f'=0$, No Outliers)} 
    \label{fig:improvement}
\end{figure}

\begin{figure}
    \centering
    \includegraphics[width=0.42\textwidth]{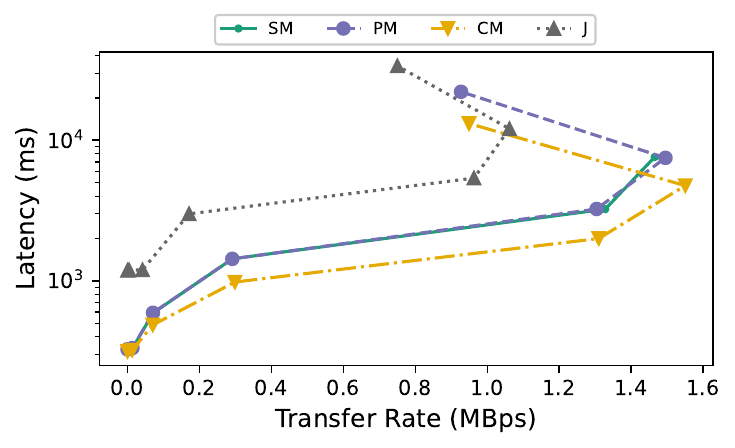}
    \caption{Throughput vs Latency ($n=200$, $f'=0$, $p\le9$MB)}
    \label{fig:throughput-v-latency}  
\end{figure}

\subsection{Fallback Path Evaluation}
We subsequently further evaluated the impact of pipelining along with \vl, reorg resilience and optimistic responsiveness, by running all protocols with a fixed $n$, $f'$, $p$ (i.e., block payload size) and $\Delta$ under three different fair LSO/LCO leader schedules. We chose $n = 100$, $f' = 33$ and $p = 0$ to maximize the impact of the quadratic steady-state complexity of our protocols without risking a repeat of the outliers seen in the $n = 200$, $f' = 0$ experiments. We also chose $\Delta = 500ms$, a somewhat-conservative value (per~\Cref{tab:observed_latency}) that still ensured that each protocol would make it through several iterations of the leader schedules within the five minute duration of each run. As for the leader schedules, the first (\Best) had all honest nodes followed by all byzantine nodes, representing the best case for non-reorg-resilient and pipelined protocols. The second (\WorstMoonshot) had honest-then-byzantine leaders for $2f'$ views, followed by honest leaders for the remaining $n-2f'$ views, representing the worst case for reorg resilient, pipelined protocols. The third (\WorstJolteon) repeated two-honest-then-byzantine for $3f'$ views, followed by the remaining $n-3f'$ honest, representing the worst case for non-reorg resilient, pipelined protocols.

\begin{figure*} 
    \centering
    \subfloat[Performance Overview\label{fig:fallback_overview}]{%
       \includegraphics[width=0.33\linewidth]{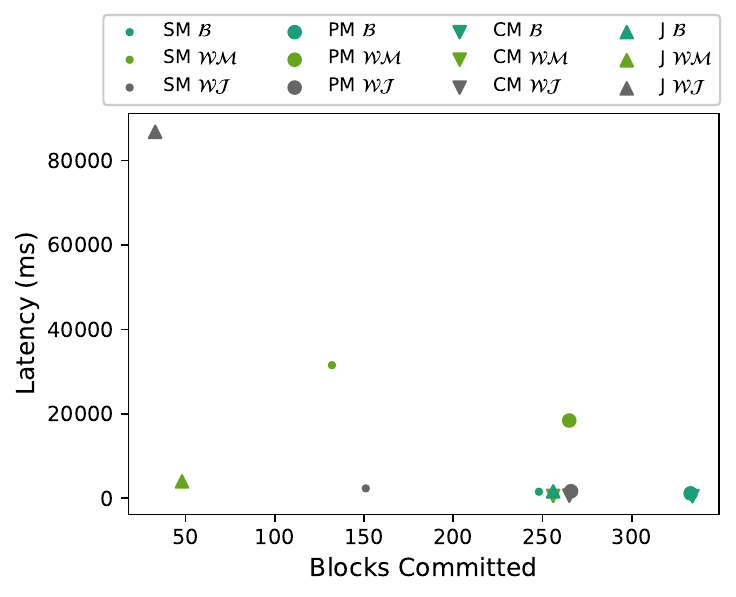}%
    }%
    \hfill
    \subfloat[Performance Overview (log scaled)\label{fig:fallback_overview_2}]{%
        \includegraphics[width=0.33\linewidth]{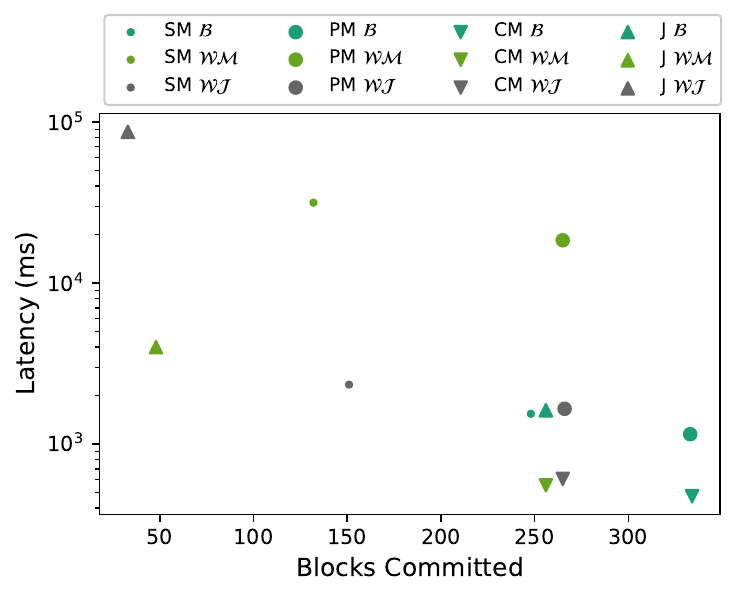}%
    }%
    \hfill
    \subfloat[Improvements vs. Jolteon\label{fig:fallback_improvement}]{%
        \includegraphics[width=0.33\linewidth]{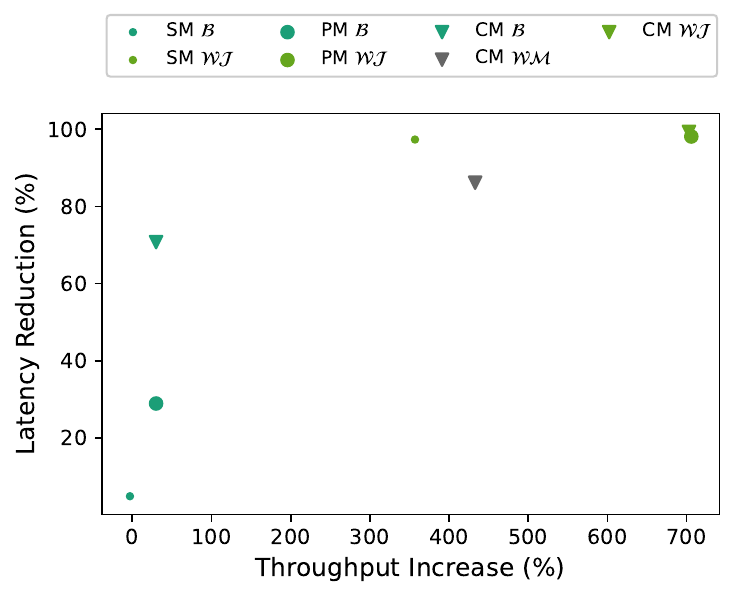}%
    }%
    \hfill
  \caption{Performance comparison at $n=100$, $f'=33$ and $p=0$}
  \label{fig1}
\end{figure*}

As shown in~\Cref{fig:fallback_overview,fig:fallback_overview_2}, Jolteon's performance degrades enormously in the presence of failures due to its lack of reorg resilience. This is evident by the difference in its results for \Best and \WorstJolteon, with the former producing approximately seven times higher throughput and fifty times lower latency than the latter. The pipelined nature of \simple and \pipelined likewise caused a significant reduction in latency between the worst (\WorstMoonshot) and best case (\Best) leader schedules for these protocols. \simple's $2\Delta$ wait after a failed leader (i.e. lack of Optimistic Responsiveness) caused its performance to vary more significantly than \pipelined, while its longer view length caused a substantial decrease in throughput.

As shown by their absence from~\Cref{fig:fallback_improvement}, both \simple and \pipelined failed to improve over Jolteon under \WorstMoonshot. More precisely, although they both produced a several-fold increase in throughput compared to Jolteon, Jolteon produced much lower latency. Both of these results were a side-effect of reorg resilience: Both \moonshot protocols committed all blocks proposed by honest leaders with Byzantine successors under this schedule, but only after a significant delay. Comparatively, Jolteon lost all such blocks due to lacking this property, with only the block of the final honest leader in the schedule being committed with a delay. Since Jolteon commits $n-2f'$ out of every $n$ blocks under this schedule, its relative improvement in block commit latency should increase proportionally to $n$, while its relative throughput should similarly decrease. However, we note that in this case reduced block commit latency at the cost of decreased throughput should be considered an undesirable trade-off as it does not imply a reduction in transaction commit latency.

Finally, \commit performed consistently well regardless of the leader schedule due to its explicit pre-commit phase, which denies the adversary any power to delay the commit of honest blocks. Notably, as shown in~\Cref{fig:fallback_improvement}, it produced around \emph{eight times higher throughput} and more than \emph{two orders of magnitude lower latency} than Jolteon under \WorstJolteon. Overall, then, \commit produced superior performance in both the happy path and in the presence of failures, making it a prime candidate for application in modern low-latency public blockchains.

\section{Related Work}
\label{sec:related_work}
There has been a long line of work towards designing efficient BFT SMR protocols for partially synchronous networks \cite{castro1999practical,abd2005fault,martin2006fast,kotla2007zyzzyva,buchman2016tendermint,yin2019hotstuff,chan2018pala,gueta2019sbft,jalalzai2020fast,abraham2021good,camenisch2022internet,gelashvili2022jolteon,danezis2022narwhal,civit2022byzantine,spiegelman2022bullshark,giridharan2023beegees,spiegelman2023shoal,malkhi2023hotstuff,malkhi2023bbca,keidar2023cordial,chan2023simplex}. Our work contributes to this effort by introducing the first CRL protocols to obtain both \mvbp $= \delta$ and \mcl $= 3\delta$. Our protocols further provide reorg resilience, improving their recovery time after a failed leader compared to prior chain-based works that fail to achieve this property. This is especially true of both \pipelined and \commit, which also have low \vl and are optimistically responsive. These properties come at the cost of $O(n^2)$ steady-state communication complexity, making our protocols less performant in this metric compared to vote-aggregator-based protocols like HotStuff. However, as shown in~\Cref{sec:eval}, this trade-off is worthwhile in many settings. We presented a brief comparison between our protocols and other recent works in~\Cref{sec:intro}. We now undertake a more thorough review.

\paragraph{Early works.} PBFT~\cite{castro1999practical} was the first practical BFT SMR protocol, achieving \mcl $= 3\delta$ at the cost of $O(n^2)$ steady-state communication. PBFT's slot-based nature complicated its view change, leading it to only rotate leaders after a failure---an approach that allows proposal frequency to be reduced below $\delta$, but precludes fairness. Much later, Tendermint~\cite{buchman2016tendermint} combined the steady-state and view-change sub-protocols into a unified protocol for the LSO setting, resulting in a simpler protocol than PBFT at the cost of an $\Omega(\Delta)$ wait before every new view at the same height, thus sacrificing optimistic responsiveness. HotStuff~\cite{yin2019hotstuff} formalized the notion of optimistic responsiveness and improved upon Tendermint both by implementing this property and being the first protocol to obtain linear ($O(n)$) communication complexity in both its steady-state and view-change phases (in the presence of an abstract pacemaker for view synchronization). To our knowledge, it was also the first protocol to implement block chaining. 

\paragraph{Linear protocols.} Like HotStuff, many other chain-based protocols~\cite{gueta2019sbft,jalalzai2020fast,gelashvili2022jolteon,sui2022marlin,malkhi2023hotstuff} have focused on minimising communication complexity, with some achieving linearity only in their steady states and others during their view-change phases as well. Recently, some~\cite{civit2022byzantine,malkhi2023hotstuff} have even achieved amortized-linear view synchronization. In all cases though, these protocols obtain steady-state linearity through the use of a designated vote-aggregator node. As we previously observed, this naturally increases their \mcl, \mvbp and \vl relative to our protocols, and precludes reorg resilience when the aggregator is not the original proposer. Moreover, while most nodes incur a steady-state complexity of $O(1)$ in these protocols, the proposer must still send and the aggregator must still receive, $O(n)$ messages. This imbalance means that these protocols under-utilize the available bandwidth in the point-to-point CRL setting (in which there should be no choke-points in the network and each node should have similar capabilities).

\paragraph{Non-linear pipelined chain-based protocols.} PaLa~\cite{chan2018pala} is a pipelined CRL protocol with \mcl $= 4\delta$ and \mvbp $= 2\delta$. While this improves upon the commit latencies of linear pipelined protocols with \mvbp $= 2\delta$, like~\cite{gelashvili2022jolteon}, PaLa achieves this result at the cost of $O(n^2)$ communication complexity in its steady state. Accordingly, PaLa is sub-optimal in all three properties.

\paragraph{Non-pipelined chain-based protocols.} Similar to PaLa, ICC~\cite{camenisch2022internet} incurs $O(n^2)$ steady-state communication complexity. However, this protocol eschews pipelining, allowing it to achieve \mcl $= 3\delta$ through the use of an explicit second round of voting for each block. Even so, it lacks reorg resilience and its \mvbp of $2\delta$ and \vl of $4\Delta$ make it less efficient in these metrics than our protocols. Simplex~\cite{chan2023simplex} obtains the same \mcl and \mvbp with \vl $= 3\Delta$, however, it claims responsiveness only when all nodes are honest. Additionally, its requirement that a leader must send the entire certified blockchain along with its proposal makes its communication complexity proportional to size of the blockchain and thus unbounded, rendering it impractical.

\paragraph{Apollo~\cite{bhat2023unique}.} Apollo obtains \mvbp $= \delta$ at the cost of a \mcl $= (f+1)\delta$ and assuming a synchronous communication model.

\paragraph{DAG-based protocols.} DAG-based consensus protocols like~\cite{spiegelman2023shoal,spiegelman2022bullshark,malkhi2023bbca} and \cite{keidar2023cordial} focus on improving block throughput. While they naturally produce and commit more blocks over a given interval than chain-based protocols by virtue of having all nodes propose in each step, they incur $O(n^3)$ communication in doing so. While recent protocols~\cite{keidar2023cordial,malkhi2023bbca} in this setting have achieved \mvbp $= \delta$, and \mcl $= 3\delta$ for blocks proposed by the leader, they require at least $4\delta$ to commit blocks proposed by other nodes. Consequently, since most blocks committed by these protocols are non-leader blocks, their average block commit latency is still higher than our protocols. Moreover, since these protocols use pipelining, each $\delta$ corresponds to one $\beta$ under our model from~\Cref{sec:commit-moonshot}, meaning that these latencies become even more significant relative to our protocols as blocks become larger.

\paragraph{Inspiration for future work.} 
\moonshot may be further optimized by applying insights from other works. For example, a related line of works~\cite{abd2005fault,kotla2007zyzzyva,martin2006fast,abraham2021good,gueta2019sbft} achieve $\lambda = 2\delta$ via optimistic commits when $n \ge 5f-1$. Similarly, works such as \cite{efficientBFT} and \cite{damysus} have leveraged trusted execution environments to limit the power of the adversary, enabling consensus when $n \ge 2f+1$. Giridharan et al. also recently proposed BeeGees~\cite{giridharan2023beegees}, a pipelined CRL protocol that is able to commit without requiring consecutive honest leaders. Defining variants of \moonshot that leverage these optimizations represents an interesting direction for future work.

\section{Conclusion}
We presented the first chain-based rotating-leader BFT protocols for the partially synchronous network model with \mvbp $= \delta$ and \mcl $= 3\delta$. All three of our protocols outperformed the previous state-of-the-art CRL protocol, Jolteon, both in the presence of failures and in failure-free scenarios. \pipelined consistently outperformed Jolteon, showing both the value of low \mvbp and reorg resilience, and that the cost of obtaining linear communication practically outweighs its benefits in many settings. Likewise, \commit equalled or outperformed \pipelined in all experiments, showing that pipelining is counter-productive in the presence of failures (without further optimization) and when blocks are large.

\section{Acknowledgments}
We thank Aniket Kate, Praveen Manjunatha, Kartik Nayak, Srivatsan Ravi and David Yang for helpful discussions related to this paper.

\Urlmuskip=0mu plus 1mu\relax
\bibliographystyle{plain}
\bibliography{references}

\begin{thebibliography}{10}

\bibitem{abd2005fault}
Michael Abd-El-Malek, Gregory~R Ganger, Garth~R Goodson, Michael~K Reiter, and
  Jay~J Wylie.
\newblock Fault-scalable byzantine fault-tolerant services.
\newblock {\em ACM SIGOPS Operating Systems Review}, 39(5):59--74, 2005.

\bibitem{abraham2020sync}
Ittai Abraham, Dahlia Malkhi, Kartik Nayak, Ling Ren, and Maofan Yin.
\newblock Sync hotstuff: Simple and practical synchronous state machine
  replication.
\newblock In {\em IEEE S\&P}, pages 106--118. IEEE, 2020.

\bibitem{abraham2021good}
Ittai Abraham, Kartik Nayak, Ling Ren, and Zhuolun Xiang.
\newblock Good-case latency of byzantine broadcast: A complete categorization.
\newblock In {\em PODC}, pages 331--341, 2021.

\bibitem{abraham2022optimal}
Ittai Abraham, Kartik Nayak, and Nibesh Shrestha.
\newblock Optimal good-case latency for rotating leader synchronous bft.
\newblock In {\em OPODIS}, 2022.

\bibitem{bhat2023unique}
Adithya Bhat, Akhil Bandarupalli, Saurabh Bagchi, Aniket Kate, and Michael
  Reiter.
\newblock Unique chain rule and its applications.
\newblock In {\em FC}, 2023.

\bibitem{eth-validators}
bitfly gmbh.
\newblock Validators chart.
\newblock \url{https://beaconcha.in/charts/validators}.
\newblock [Online; accessed 04-December-2023].

\bibitem{blum2023analyzing}
Erica Blum, Derek Leung, Julian Loss, Jonathan Katz, and Tal Rabin.
\newblock Analyzing the real-world security of the algorand blockchain.
\newblock In {\em CCS}, pages 830--844, 2023.

\bibitem{boneh2004short}
Dan Boneh, Ben Lynn, and Hovav Shacham.
\newblock Short signatures from the weil pairing.
\newblock {\em Journal of cryptology}, 17:297--319, 2004.

\bibitem{bracha1987asynchronous}
Gabriel Bracha.
\newblock Asynchronous byzantine agreement protocols.
\newblock {\em Information and Computation}, 75(2):130--143, 1987.

\bibitem{buchman2016tendermint}
Ethan Buchman.
\newblock {\em Tendermint: Byzantine fault tolerance in the age of
  blockchains}.
\newblock PhD thesis, University of Guelph, 2016.

\bibitem{camenisch2022internet}
Jan Camenisch, Manu Drijvers, Timo Hanke, Yvonne-Anne Pignolet, Victor Shoup,
  and Dominic Williams.
\newblock Internet computer consensus.
\newblock In {\em PODC}, pages 81--91, 2022.

\bibitem{castro1999practical}
Miguel Castro, Barbara Liskov, et~al.
\newblock Practical byzantine fault tolerance.
\newblock In {\em OSDI}, volume~99, pages 173--186, 1999.

\bibitem{chan2023simplex}
Benjamin~Y Chan and Rafael Pass.
\newblock Simplex consensus: A simple and fast consensus protocol.
\newblock In {\em TCC 2023 (To appear)}. Springer, 2023.

\bibitem{chan2018pala}
TH~Hubert Chan, Rafael Pass, and Elaine Shi.
\newblock Pala: A simple partially synchronous blockchain.
\newblock {\em Cryptology ePrint Archive}, 2018.

\bibitem{civit2022byzantine}
Pierre Civit, Muhammad~Ayaz Dzulfikar, Seth Gilbert, Vincent Gramoli, Rachid
  Guerraoui, Jovan Komatovic, and Manuel Vidigueira.
\newblock Byzantine consensus is $\theta$ (n$^2$): The dolev-reischuk bound is
  tight even in partial synchrony!
\newblock In {\em DISC}, 2022.

\bibitem{danezis2022narwhal}
George Danezis, Lefteris Kokoris-Kogias, Alberto Sonnino, and Alexander
  Spiegelman.
\newblock Narwhal and tusk: a dag-based mempool and efficient bft consensus.
\newblock In {\em Eurosys}, pages 34--50, 2022.

\bibitem{damysus}
J\'{e}r\'{e}mie Decouchant, David Kozhaya, Vincent Rahli, and Jiangshan Yu.
\newblock Damysus: streamlined bft consensus leveraging trusted components.
\newblock In {\em Proceedings of the Seventeenth European Conference on
  Computer Systems}, EuroSys '22, page 1–16, New York, NY, USA, 2022.
  Association for Computing Machinery.

\bibitem{dwork1988consensus}
Cynthia Dwork, Nancy Lynch, and Larry Stockmeyer.
\newblock Consensus in the presence of partial synchrony.
\newblock {\em Journal of the ACM (JACM)}, 35(2):288--323, 1988.

\bibitem{aptos}
Aptos Foundation.
\newblock Aptos: The world's most production-ready blockchain.
\newblock \url{https://aptosfoundation.org/}.
\newblock [Online; accessed 28-February-2024].

\bibitem{gelashvili2022jolteon}
Rati Gelashvili, Lefteris Kokoris-Kogias, Alberto Sonnino, Alexander
  Spiegelman, and Zhuolun Xiang.
\newblock Jolteon and ditto: Network-adaptive efficient consensus with
  asynchronous fallback.
\newblock In {\em FC}, pages 296--315, 2022.

\bibitem{giridharan2023beegees}
Neil Giridharan, Florian Suri-Payer, Matthew Ding, Heidi Howard, Ittai Abraham,
  and Natacha Crooks.
\newblock Beegees: stayin'alive in chained bft.
\newblock In {\em PODC}, pages 233--243, 2023.

\bibitem{gueta2019sbft}
Guy~Golan Gueta, Ittai Abraham, Shelly Grossman, Dahlia Malkhi, Benny Pinkas,
  Michael Reiter, Dragos-Adrian Seredinschi, Orr Tamir, and Alin Tomescu.
\newblock Sbft: A scalable and decentralized trust infrastructure.
\newblock In {\em DSN}, pages 568--580. IEEE, 2019.

\bibitem{jalalzai2020fast}
Mohammad~M Jalalzai, Jianyu Niu, Chen Feng, and Fangyu Gai.
\newblock Fast-hotstuff: A fast and resilient hotstuff protocol.
\newblock {\em arXiv preprint arXiv:2010.11454}, 2020.

\bibitem{keidar2023cordial}
Idit Keidar, Oded Naor, Ouri Poupko, and Ehud Shapiro.
\newblock Cordial miners: Fast and efficient consensus for every eventuality.
\newblock In {\em DISC}, 2023.

\bibitem{kotla2007zyzzyva}
Ramakrishna Kotla, Lorenzo Alvisi, Mike Dahlin, Allen Clement, and Edmund Wong.
\newblock Zyzzyva: speculative byzantine fault tolerance.
\newblock In {\em SOSP}, pages 45--58, 2007.

\bibitem{aptos-validators}
Aptos Labs.
\newblock Validators.
\newblock \url{https://explorer.aptoslabs.com/validators/all?network=mainnet}.
\newblock [Online; accessed 04-December-2023].

\bibitem{malkhi2023hotstuff}
Dahlia Malkhi and Kartik Nayak.
\newblock Hotstuff-2: Optimal two-phase responsive bft.
\newblock {\em Cryptology ePrint Archive}, 2023.

\bibitem{malkhi2023bbca}
Dahlia Malkhi, Chrysoula Stathakopoulou, and Maofan Yin.
\newblock Bbca-chain: One-message, low latency bft consensus on a dag.
\newblock {\em arXiv preprint arXiv:2310.06335}, 2023.

\bibitem{martin2006fast}
J-P Martin and Lorenzo Alvisi.
\newblock Fast byzantine consensus.
\newblock {\em TDSC}, 3(3):202--215, 2006.

\bibitem{pass2017hybrid}
Rafael Pass and Elaine Shi.
\newblock Hybrid consensus: Efficient consensus in the permissionless model.
\newblock In {\em DISC}, page~6, 2017.

\bibitem{narwhal-hs}
Facebook Research.
\newblock Narwhal-hotstuff github repository.
\newblock \url{https://github.com/facebookresearch/narwhal/tree/narwhal-hs}.
\newblock [Online; accessed 22-January-2023].

\bibitem{shi2019streamlined}
Elaine Shi.
\newblock Streamlined blockchains: A simple and elegant approach (a tutorial
  and survey).
\newblock In {\em ASIACRYPT}, pages 3--17. Springer, 2019.

\bibitem{spiegelman2023shoal}
Alexander Spiegelman, Balaji Aurn, Rati Gelashvili, and Zekun Li.
\newblock Shoal: Improving dag-bft latency and robustness.
\newblock {\em arXiv preprint arXiv:2306.03058}, 2023.

\bibitem{spiegelman2022bullshark}
Alexander Spiegelman, Neil Giridharan, Alberto Sonnino, and Lefteris
  Kokoris-Kogias.
\newblock Bullshark: Dag bft protocols made practical.
\newblock In {\em CCS}, pages 2705--2718, 2022.

\bibitem{sui2022marlin}
Xiao Sui, Sisi Duan, and Haibin Zhang.
\newblock Marlin: Two-phase bft with linearity.
\newblock In {\em DSN}, pages 54--66. IEEE, 2022.

\bibitem{efficientBFT}
Giuliana~Santos Veronese, Miguel Correia, Alysson~Neves Bessani, Lau~Cheuk
  Lung, and Paulo Verissimo.
\newblock Efficient byzantine fault-tolerance.
\newblock {\em IEEE Transactions on Computers}, 62(1):16--30, 2013.

\bibitem{yin2019hotstuff}
Maofan Yin, Dahlia Malkhi, Michael~K Reiter, Guy~Golan Gueta, and Ittai
  Abraham.
\newblock Hotstuff: Bft consensus with linearity and responsiveness.
\newblock In {\em PODC}, pages 347--356, 2019.

\end{thebibliography}

\appendices
\section{\simple Security Analysis}\label{sec:simple-proof}
We now present formal proofs that \simple satisfies the safety and liveness properties of SMR, and reorg resilience.

\begin{claim}[Quorum Intersection]\label{clm:quorum-intersection}
    Given any two quorums $Q_1$ and $Q_2$ drawn from $\mathcal{V}$, $Q_1$ and $Q_2$ have at least one honest node in common.
\end{claim}
\begin{proof}
    According to the definition given in~\Cref{sec:preliminaries}, when $n = 3f+1$ a quorum contains $2f+1$ distinct members of $\mathcal{V}$.
    Therefore, $Q_1$ and $Q_2$ must have at least $f+1$ members in common. 
    Thus, because $\mathcal{V}$ contains only $f$ Byzantine nodes, at least one of these shared nodes must be honest.
\end{proof}

\begin{claim}[Honest Majority Intersection]\label{clm:honest-majority-intersection}
    Given any two sets $H_1$ and $H_2$ of at least $f+1$ honest nodes drawn from $\mathcal{V}$, $H_1$ and $H_2$ have at least one honest node in common.
\end{claim}
\begin{proof}
    According to the definition given in~\Cref{sec:preliminaries}, when $n = 3f+1$, $\mathcal{V}$ contains $2f+1$ honest nodes.
    Therefore, since $H_1$ and $H_2$ both contain at least $f+1$ honest nodes, they must have at least one node in common.
\end{proof}

\begin{lemma}\label{lem:view-safety}
If $\Cert_v(B_k)$ and $\Cert_v(B_l)$ exist then $B_k = B_l$.
\end{lemma}
\begin{proof}
Suppose, for the sake of contradiction, that $\Cert_v(B_k)$ and $\Cert_v(B_l)$ exist but $B_k \ne B_l$.
By the definition of a block certificate given in~\Cref{sec:preliminaries},
the existence of $\Cert_v(B_k)$ implies that at least $2f+1$ nodes voted for $B_k$ in view $v$. 
Likewise, the existence of $\Cert_v(B_l)$ implies that the same number of nodes also voted for $B_l$ in $v$. 
By~\Cref{clm:quorum-intersection}, this implies that at least one honest node voted for both $B_k$ and $B_l$ in $v$. 
However, this violates the rules for voting, which allow a node to vote only once per view, contradicting the original assumption.
\end{proof}

\begin{claim}\label{clm:honest-normal-proposal-cert}
    If an honest node, say $\node{i}$, votes for $\sig{\Propose, B_l, \Cert_{v'}(B_h), v}$, then $v' < v$.
\end{claim}
\begin{proof}
    Since the view advancement rule takes priority over the voting rule, if $v' \ge v$ then $\node{i}$ would have entered $v' + 1 > v$ before voting for $\sig{\Propose, B_l, \Cert_{v'}(B_h), v}$, making it ineligible to vote for this proposal.
\end{proof}

\begin{lemma}\label{lem:uniq-xtn-base}
If an honest node, say $\node{i}$, directly commits a block $B_k$ that was certified for view $v$ and $\Cert_{v'}(B_{k'})$ exists such that $v' = v$ or $v' = v+1$, then $B_{k'}$ extends $B_k$.
\end{lemma}
\begin{proof}
    If $v' = v$ then, by~\Cref{lem:view-safety}, $B_{k'} = B_k$ and thus, per the definition of block extension given in~\Cref{sec:preliminaries}, $B_{k'}$ extends $B_k$.
    Alternatively, if $v' = v+1$ then, by the direct commit rule, $\node{i}$ must have observed $\Cert_v(B_k)$ and $\Cert_{v+1}(B_{k+1})$ with $B_{k+1}$ extending $B_k$.
    Additionally, by~\Cref{lem:view-safety}, $\Cert_{v+1}(B_{k+1})$ is the only block certificate that can exist for $v+1$. 
    Thus, $B_{k'} = B_{k+1}$, so $B_{k'}$ extends $B_k$.
\end{proof}

\begin{lemma}[Unique Extensibility]\label{lem:uniq-xtn}
If an honest node, say $\node{i}$, directly commits a block $B_k$ that was certified for view $v$ and $\Cert_{v'}(B_{k'})$ exists such that $v' \ge v$, then $B_{k'}$ extends $B_k$.
\end{lemma}
\begin{proof}
We complete this proof by strong induction on $v'$, however, \Cref{lem:uniq-xtn-base} covers the base cases ($v' = v$ and $v' = v+1$) so we proceed directly with the inductive step.

\paragraph{Inductive step: $v' > v+1$.} For our induction hypothesis, we assume that the lemma holds up to view $v'-1$. 
That is, we assume that every $\Cert_{v^*}(B_{k^*})$ with $v \le v^* < v'$ extends $B_k$.
We use this assumption to prove that it also holds for $v'$. 
We first observe that the existence of $\Cert_{v'}(B_{k'})$ implies that a set $H_1$ of at least $f+1$ honest nodes voted for $B_{k'}$ in view $v'$.
If any of these nodes voted using the optimistic vote rule, then they must have been locked on $\Cert_{v'-1}(B_{k'-1})$ and $B_{k'}$ must extend $B_{k'-1}$.
Therefore, since by the induction hypothesis $B_{k'-1}$ extends $B_k$, $B_{k'}$ also extends $B_k$.
Alternatively, if no honest node used the optimistic vote rule to vote for $B_{k'}$ then all members of $H_1$ must have used the normal vote rule to vote for $B_{k'}$. 
Therefore, they must have received $\sig{\Propose, B_{k'}, \Cert_{v''}(B_{k''}), v'}$ such that $B_{k'}$ extended $B_{k''}$ and $\Cert_{v''}(B_{k''})$ ranked equal to or higher than their respective locks.
Moreover, by~\Cref{clm:honest-normal-proposal-cert}, $v'' < v'$.
We now show that $v'' \ge v$.

Recall that we know from the commit rule that $\Cert_{v+1}(B_{k+1})$ exists and that $B_{k+1}$ extends $B_k$.
Therefore, a set, say $H_2$, of at least $f+1$ honest nodes must have voted for $B_{k+1}$ in view $v+1$.
Furthermore, by the vote rules, they must have done this after receiving $\Cert_v(B_k)$ and therefore would have locked $\Cert_v(B_k)$ or a higher ranked block certificate upon advancing to a new view.
By~\Cref{clm:honest-majority-intersection}, $H_1$ and $H_2$ must have at least one member, say $P_i$, in common.
Since the view advancement rule ensures that $P_i$ never decreases its local view, it must have voted for $B_{k'}$ (in $v'$) after $B_{k+1}$ (in $v+1$) and thus must have been locked on $\Cert_v(B_k)$ or higher upon doing so.
Hence, since the normal vote rule ensures that $\Cert_{v''}(B_{k''}) \ge \Lock_i$ for $\node{i}$, by the block certificate ranking rule, $v'' \ge v$.
Hence, since $v \le v'' < v'$, by the induction hypothesis, $B_{k''}$ extends $B_k$, so $B_{k'}$ also extends $B_k$.
\end{proof}

\begin{theorem}[Safety]\label{thm:safety}
Honest nodes do not commit different values at the same log position.
\end{theorem}
\begin{proof}
We show that if two honest nodes $P_i$ and $P_j$ commit $B_k$ and $B'_k$, then $B_k = B'_k$. This fact together with the assumptions mentioned along with~\Cref{def:smr} in~\Cref{sec:preliminaries}, is sufficient to achieve safety.

Suppose, for the sake of contradiction, that $P_i$ and $P_j$ commit $B_k$ and $B'_k$ but $B_k \ne B'_k$.
By the indirect commit rule, $P_i$ and $P_j$ must do so as a result of respectively directly committing blocks $B_l$ and $B_m$ such that $B_l$ extends $B_k$ and $l \ge k$, and $B_m$ extends $B'_k$ and $m \ge k$.
Thus, by~\Cref{lem:uniq-xtn}, either $v \le v'$ and $B_m$ extends $B_l$, or $v \ge v'$ and $B_l$ extends $B_m$.
Therefore, since $B_l$ and $B_m$ are a part of the same chain and because each block in the chain has exactly one parent, $B_k = B'_k$.
\end{proof}

\begin{claim}\label{clm:all-enter} 
Let $t_g$ denote GST. If the first honest node to enter view $v$ does so at time $t$, then all honest nodes enter $v$ or higher by $\max(t_g, t) + \Delta$.
\end{claim}
\begin{proof}
Let $\node{i}$ be the first honest node to enter $v$. 
By the view advancement rule, it must have entered $v$ via either $\Cert_{v-1}$ or $\TimeoutCert_{v-1}$ and must have multicasted this certificate upon doing so. 
Therefore, since messages sent by honest nodes arrive within $\Delta$ time after GST, all honest nodes will receive this certificate by $\max(t_g, t) + \Delta$ and thus will enter $v$ if they have not already entered a higher view.
\end{proof}

\begin{lemma}\label{lem:view-progress}
All honest nodes keep entering increasing views.
\end{lemma}    
\begin{proof}
    Suppose, for the sake of contradiction, that at least one honest node, say $\node{i}$, becomes stuck in view $v$ and let $v'$ be the highest view of any honest node at any time.
    If $v' > v$ then~\Cref{clm:all-enter} shows that $\node{i}$ will enter $v'$ or higher, contradicting the assumption that it becomes stuck in $v$.
    Otherwise, if $v' = v$ then since this implies that no honest node ever enters a view higher than $v$ and because~\Cref{clm:all-enter} shows that all honest nodes will enter $v$, they must all become stuck there.
    However, by the view advancement and timeout rules, these nodes will all eventually multicast $\TimeoutMessage_v$ and thus will all be able to construct $\TimeoutCert_v$ and enter $v+1$, contradicting the conclusion that they must become stuck in $v$.
\end{proof}

\begin{claim}
    \label{clm:sequential-progress}
    If an honest node enters view $v$ then at least $f+1$ honest nodes must have already entered $v-1$.
\end{claim}
\begin{proof}
    The view advancement rule requires an honest node to observe either $\Cert_{v-1}$ or $\TimeoutCert_{v-1}$ in order to enter $v$.
    Therefore, at least $f+1$ honest nodes must send the corresponding messages.
    Moreover, by the vote and timeout rules, they must do so whilst in $v$.
    Thus, in either case, an honest node can only enter $v$ if at least $f+1$ honest nodes have already entered $v-1$.
\end{proof}

\begin{lemma}\label{lem:locking}
If the first honest node to enter view $v$ does so after GST and $L_v$ is honest, then all honest nodes receive $\Cert_v(B_k)$ for some block $B_k$ proposed by $L_v$, and at least $f+1$ of them lock this certificate while entering $v+1$.
\end{lemma}
\begin{proof}
By~\Cref{lem:view-safety}, only one block can become certified for a given view.
Thus, if $\Cert_v(B_k)$ exists then any node that receives a view $v$ block certificate must receive $\Cert_v(B_k)$.
We assume this fact in the remainder of the proof.

Let $t$ be the time when the first honest node enters view $v$.
Because honest nodes only send $\TimeoutMessage_v$ either after receiving $f+1$ such messages from unique senders, or upon their view timers expiring, no honest node will send $\TimeoutMessage_v$ until $t + 5\Delta$.
Moreover, since by~\Cref{clm:sequential-progress} no honest node can enter a view greater than $v$ until at least $f+1$ honest nodes enter $v$, neither can any honest node send a timeout message for a view greater than $v$ before this time.
Thus, no honest node can enter a view greater than $v$ via a timeout certificate before $t + 5\Delta$.
Consequently, since by the same claim no honest node can enter a view greater than $v+1$ unless at least $f+1$ honest nodes first enter $v+1$, if any honest node enters a view greater than $v+1$ before $t + 5\Delta$ then, by the view advancement rule, at least $f+1$ honest nodes must have locked and multicasted $\Cert_v(B_k)$ upon entering $v+1$.
Alternatively, if any honest node enters $v+1$ before $t + 4\Delta$ then, by the view advancement rule, it will multicast $\Cert_v$ upon doing so, which all nodes will receive before $t + 5\Delta$.
Therefore, either all honest nodes enter $v+1$ and lock $\Cert_v$ before $t + 5\Delta$, or some honest node enters a view greater than $v+1$ before $t + 5\Delta$.
In either case, the proof is complete.

Suppose, then, both that no honest node enters a view greater than $v+1$ before $t + 5\Delta$ and that no honest node enters $v+1$ before $t + 4\Delta$.
Therefore, by~\Cref{clm:all-enter}, all honest nodes will enter $v$ before $t + \Delta$.
If $L_v$ enters $v$ via $\Cert_{v-1}(B_h)$, then it will multicast $\sig{\Propose, B_{h+1}, \Cert_{v-1}(B_h), v}$ with $B_{h+1}$ extending $B_h$, which all honest nodes will receive before $t + 2\Delta$.
Therefore, if all honest nodes vote for $B_{h+1}$ no later than the time that they receive this proposal then they will all receive $\Cert_v(B_{h+1})$ before $t + 3\Delta$.
Thus, by~\Cref{clm:sequential-progress}, at least $f+1$ of them will enter $v+1$ via this certificate and will subsequently lock it, completing the proof.
Otherwise, some honest node, say $\node{j}$, must fail to vote for $B_{h+1}$ before $t + 2\Delta$.
However, since we have already considered the case where any honest node enters a view greater than $v$ before $t + 4\Delta$, $\node{j}$ cannot have $\Lock_i > \Cert_{v-1}(B_h)$ when it attempts to vote for $B_{h+1}$.
Therefore, since $L_v$ will ensure that $B_{h+1}$ extends $B_h$, $\node{j}$ can only have failed to vote for $B_{h+1}$ if it had already voted in view $v$.
However, since $L_v$ is honest it will only create a single normal proposal, so $\node{j}$ must have voted for an optimistic proposal containing some block $B_k$.
However, since~\Cref{lem:view-safety} shows that only one block can become certified for $v-1$, by the optimistic vote rule, $B_k$ extends $B_h$.
Moreover, since we have defined block payloads as being fixed for a given view, because $L_v$ is honest, $B_k$ must also have the same payload as $B_{h+1}$.
Thus, $B_k = B_{h+1}$, contradicting the conclusion that $\node{j}$ must have failed to vote for $B_k$ and completing the proof.

Otherwise, if $L_v$ enters $v$ via $\TimeoutCert_{v-1}$ then it will wait $2\Delta$ before proposing.
As before, this implies that all honest nodes enter $v$ before $t + \Delta$.
By the view advancement rule, any node that does so via $\TimeoutCert_{v-1}$ will unicast $\sig{\Status, v, \Lock_i}$ to $L_v$.
Similarly, any node that enters $v$ via $\Cert_{v-1}$ will multicast this certificate.
Consequently, $L_v$ will receive the highest ranked block certificate, say $\Cert_{v'}(B_h)$, known to any honest node before $t + 3\Delta$.
Thus, since $L_v$ is honest, when it proposes it will multicast a normal proposal containing a block that extends $B_h$; i.e., $\sig{\Propose, B_{h+1}, \Cert_{v'}(B_h), v}$.
All honest nodes will receive this proposal before $t + 4\Delta$.
Furthermore, if they all vote for $B_{h+1}$ before this time then they will all receive $\Cert_v(B_{h+1})$ before $t + 5\Delta$.
Thus, since we have already concluded that no honest node can enter $v+1$ or higher via a timeout certificate before this time, by~\Cref{clm:sequential-progress}, at least $f+1$ of them will lock $\Cert_v(B_{h+1})$ upon entering $v+1$.
Otherwise, some honest node, say $\node{j}$, must fail to vote for $B_{h+1}$ before $t + 4\Delta$.
However, we already know that all honest nodes will enter $v$ before $t + \Delta$ and will have $\Lock_i \le \Cert_{v'}(B_h)$ upon receiving $L_v$'s proposal, which will occur before $t + 4\Delta$ and thus before any of them can have sent $\TimeoutMessage_v$.
Moreover, as previously reasoned, $L_v$ will not create an equivocal proposal that $\node{j}$ can vote for.
Therefore, $\node{j}$ must vote for $B_{h+1}$ before $t + 4\Delta$. Thus, as before, all honest nodes will receive $\Cert_v(B_{h+1})$ before $t + 5\Delta$ and since no honest node can enter $v+1$ or higher via a timeout certificate before this time, by~\Cref{clm:sequential-progress}, at least $f+1$ of them will lock this certificate upon entering $v+1$, completing the proof.
\end{proof}

\begin{lemma}\label{lem:certified-successor}
    If the first honest node to enter view $v$ does so after GST, $L_v$ is honest and proposes a block $B_k$ that becomes certified, and $\Cert_{v+1}(B_l)$ exists, then $B_l$ directly extends $B_k$.
\end{lemma}
\begin{proof}
    By~\Cref{lem:locking}, a set $H_1$ of at least $f+1$ honest nodes lock $\Cert_v(B_k)$ while entering $v+1$. 
    Furthermore, $\Cert_{v+1}(B_l)$ can only exist if a set $H_2$ of at least $f+1$ honest nodes vote for $B_l$ in view $v+1$. 
    By~\Cref{clm:honest-majority-intersection}, $H_1$ and $H_2$ must have at least one node, say $\node{i}$, in common.
    Thus, since the optimistic vote rule requires $\node{i}$ to be locked on the parent of $B_l$, if $\node{i}$ votes for an optimistic proposal containing $B_l$ then $B_l$ must directly extend $B_k$.
    Alternatively, $\node{i}$ must vote for $\sig{\Propose, B_l, \Cert_{v'}(B_h), v+1}$.
    Thus, since by~\Cref{lem:view-safety} and~\Cref{clm:honest-normal-proposal-cert} $\Cert_{v'}(B_h) = \Cert_v(B_k)$, $B_l$ must directly extend $B_k$.
\end{proof}

\begin{theorem}[Liveness]\label{thm:liveness}
Each client request is eventually committed by all honest nodes.
\end{theorem}
\begin{proof}
We show that all honest nodes continue to commit new blocks to their local blockchains after GST, which, together with the assumptions mentioned along with~\Cref{def:smr} in~\Cref{sec:preliminaries}, is sufficient to achieve liveness.

By~\Cref{lem:view-progress}, all honest nodes continually enter higher views. Therefore, the protocol eventually reaches two consecutive views after GST, say $v$ and $v+1$, that have leaders $L_v$ and $L_{v+1}$ that are both honest.
By~\Cref{lem:locking}, all honest nodes will receive the same $\Cert_v$ and at least $f+1$ of them will lock it upon entering $v+1$.
Repeated application of this lemma for $L_{v+1}$ shows that all honest nodes will also receive the same $\Cert_{v+1}$.
Let the blocks certified by $\Cert_v$ and $\Cert_{v+1}$ be denoted $B_k$ and $B_l$ respectively.
\Cref{lem:certified-successor} shows that $B_l$ directly extends $B_k$.
Consequently, by the commit rule, all honest nodes will commit $B_k$ upon receiving both $\Cert_v(B_k)$ and $\Cert_{v+1}(B_l)$.
Thus, \simple commits a new block every time two consecutive, honest leaders are elected after GST.
\end{proof}

\begin{theorem}[Reorg resilience]\label{the:reorg-resilience}
If the first honest node to enter view $v$ does so after GST and $L_v$ is honest and proposes, then one of its proposed blocks, say $B_k$, becomes certified and for every $\Cert_{v'}(B_{k'})$ such that $v' \ge v$, $B_{k'}$ extends $B_k$.
\end{theorem}
\begin{proof}
By~\Cref{lem:locking}, $L_v$ produces a certified block. 
Let this block be denoted $B_k$.
We now show that for every $\Cert_v'(B_{k'})$, $B_{k'}$ extends $B_k$.

If $v' = v$ then, by~\Cref{lem:view-safety}, $B_{k'} = B_k$ and thus, per the definition of block extension given in~\Cref{sec:preliminaries}, $B_{k'}$ extends $B_k$. We now complete the proof for $v' > v$ by strong induction on $v'$, however, since~\Cref{lem:certified-successor} covers the base case ($v' = v+1$), we proceed directly with the inductive step.

\paragraph{Inductive step: $v' > v+1$.} For our induction hypothesis, we assume that the theorem holds up to view $v'-1$. 
That is, we assume that every $\Cert_{v^*}(B_{k^*})$ with $v \le v^* < v'$ extends $B_k$.
We use this assumption to prove that it also holds for $v'$. 
If any honest node votes for an optimistic proposal containing $B_{k'}$ then, by the optimistic vote rule, it must be locked on $\Cert_{v'-1}(B_{k'-1})$ such that $B_k$ extends $B_{k'-1}$.
Therefore, since by the induction hypothesis $B_{k'-1}$ extends $B_k$, $B_{k'}$ also extends $B_k$.
Otherwise, a set $H_1$ of at least $f+1$ honest nodes vote for $B_{k'}$ via the normal vote rule.
By~\Cref{lem:locking}, a set $H_2$ of at least $f+1$ honest nodes lock $\Cert_v(B_k)$ while entering $v+1$. 
By~\Cref{clm:honest-majority-intersection}, $H_1$ and $H_2$ must have at least one node, say $\node{i}$, in common.
By the normal vote rule, $\node{i}$ will only vote for $\sig{\Propose, B_{k'}, \Cert_{v''}(B_h), v'}$ when $\Cert_{v''}(B_h) \ge \Lock_i$ and $B_{k'}$ extends $B_h$.
Moreover, by~\Cref{clm:honest-normal-proposal-cert}, $v'' < v'$.
Additionally, since $\node{i}$ locks $C_v(B_k)$ upon entering $v+1$ and thus before entering $v'$, $v'' \ge v$.
Therefore, since $v \le v'' < v'$, by the induction hypothesis, $B_{k
'}$ extends $B_k$.
\end{proof}

\section{\pipelined Security Analysis}\label{sec:pipelined-proof}
We now present formal proofs that \pipelined satisfies the safety and liveness properties of SMR, and reorg resilience.

\begin{claim}\label{clm:oc-no-tc}
If $\Cert^o_{v}(B_k)$ exists then $\TimeoutCert_{v-1}$ does not exist, and vice-versa.
\end{claim}
\begin{proof}
Suppose for the sake of contradiction, that both certificates exist.
Therefore, by~\Cref{clm:quorum-intersection} and~\Cref{clm:honest-majority-intersection}, at least one honest node, say $\node{i}$, must have both sent $\sig{\OptVote, H(B_k), v}_i$ and $\sig{\Timeout, v-1, \Lock_i}_{i}$.
Furthermore, by the optimistic vote rule, $\node{i}$ must have had $\TimeoutView < v-1$ upon voting for $B_k$ and thus must have sent its timeout message for $v-1$ after its optimistic vote for $B_k$.
However, by the same rule, $\node{i}$ must have been in view $v$ when it voted for $B_k$ and hence, by the timeout rule, would have been unable to multicast $\TimeoutMessage$ messages for view $v-1$ or lower after doing so, contradicting the earlier conclusion that it must have sent $\sig{\Timeout, v-1, \Lock_i}_{i}$.
\end{proof}

\begin{claim}[Optimistic Equivalence]
    \label{clm:optimistic-equivalence}
    If $\Cert^o_{v}(B_k)$ and $\Cert^n_{v}(B_l)$ exist then $B_k = B_l$.
\end{claim}
\begin{proof}
    By~\Cref{clm:quorum-intersection}, \Cref{clm:honest-majority-intersection} and the requirement that block certificates be constructed from a quorum of votes of the same type for the same block, at least one honest node, say $\node{i}$, must have voted for both $B_k$ and $B_l$.
    By the optimistic vote rule, $\node{i}$ can only have voted for $B_k$ if it had not already voted in $v$ and thus must have voted for $B_l$ after $B_k$.
    Therefore, since the normal vote rule only allows $\node{i}$ to vote for $B_l$ if it has not already sent an optimistic vote for an equivocating block, by the definition of equivocation given in~\Cref{sec:preliminaries}$, \node{i}$ can only have voted for $B_l$ after $B_k$ if $B_l = B_k$. 
    Thus, since $\node{i}$ must have voted for both blocks, $B_l = B_k$.
\end{proof}

\begin{lemma}
    \label{lem:view-safety-2}
    If $\Cert_v(B_k)$ and $\Cert_v(B_l)$ exist then $B_k = B_l$.
\end{lemma}
\begin{proof}
    By~\Cref{clm:quorum-intersection} and \Cref{clm:honest-majority-intersection}, at least one honest node, say $\node{i}$, must have voted towards both certificates.
    There are four cases to consider:
    \begin{enumerate}
        \item When both certificates have the same type.
        \item When $\Cert_v(B_k)$ is $\Cert^n_v(B_k)$ and $\Cert_v(B_l)$ is $\Cert^f_v(B_l)$, or vice-versa.
        \item When $\Cert_v(B_k)$ is $\Cert^o_v(B_k)$ and $\Cert_v(B_l)$ is $\Cert^f_v(B_l)$, or vice-versa.
        \item When $\Cert_v(B_k)$ is $\Cert^o_v(B_k)$ and $\Cert_v(B_l)$ is $\Cert^n_v(B_l)$, or vice-versa.
    \end{enumerate}
    In the first case, since each vote rule may be triggered at most once in a given view, $\node{i}$ can only have voted towards both certificates if $B_k = B_l$.
    In the second case, because the respective vote rules prevent a node from voting if it has already voted for a proposal of the other type, $\node{i}$ cannot have voted towards both certificates, contradicting the earlier conclusion that it must have done so.
    In the third case, by the fallback vote rule, $\node{i}$ can only have voted for $B_l$ if it were justified by $\TimeoutCert_{v-1}$.
    Therefore, by Lemma~\ref{clm:oc-no-tc}, $\Cert^o_v(B_k)$ cannot exist, contradicting the assumption that it does.
    Finally, \Cref{clm:optimistic-equivalence} covers the last case.
    Thus, $B_k = B_l$.
\end{proof}

\Cref{fact:timeout-qc} follows from the lock rule, which requires a node to update $\Lock_i$ to the highest ranked QC that it has received.
\begin{fact}
    \label{fact:timeout-qc}
    If an honest node receives $\Cert_v$ then every $\TimeoutMessage$ message that it multicasts after doing so contains a block certificate with a rank $v'$ such that $v' \ge v$.
\end{fact}

\begin{claim}\label{clm:honest-normal-proposal-cert-2}
    If an honest node, say $\node{i}$, votes for $\sig{\FallbackPropose, B_{k}, \Cert_{v'}(B_h), \TimeoutCert_{v-1}, v}$, then $v' < v$.
\end{claim}
\begin{proof}
    Since the view advancement rule takes priority over the voting rule, if $v' \ge v$ then $\node{i}$ would have entered $v' + 1 > v$ before voting for $\sig{\FallbackPropose, B_{k}, \Cert_{v'}(B_h), \TimeoutCert_{v-1}, v}$, making it ineligible to vote for this proposal.
\end{proof}

\begin{lemma}[Unique Extensibility]\label{lem:uniq-xtn-pipelined}
    If an honest node, say $\node{i}$, directly commits a block $B_k$ that was certified for view $v$ and $\Cert_{v'}(B_{k'})$ exists such that $v' \ge v$, then $B_{k'}$ extends $B_k$.
\end{lemma}
\begin{proof}
    As in~\Cref{lem:uniq-xtn}, we complete this proof by strong induction on $v'$. 
    As before, \Cref{lem:uniq-xtn-base} covers the base cases ($v' = v$ and $v' = v+1$), except that~\Cref{lem:view-safety-2} needs to be invoked instead of~\Cref{lem:view-safety}.
    Accordingly, we proceed directly with the inductive step.

    \paragraph{Inductive step: $v' > v+1$.} For our induction hypothesis, we assume that the lemma holds up to view $v'-1$. 
    That is, we assume that every $\Cert_{v^*}(B_{k^*})$ with $v \le v^* < v'$ extends $B_k$.
    We use this assumption to prove that it also holds for $v'$. 
    We first observe that the existence of $\Cert_{v'}(B_{k'})$ implies that a set $H_1$ of at least $f+1$ honest nodes voted for $B_{k'}$ in view $v'$.
    If any of these nodes voted using the optimistic or normal vote rules then they must have received $\Cert_{v'-1}(B_{k'-1})$ and $B_{k'}$ must extend $B_{k'-1}$.
    Therefore, since by the induction hypothesis $B_{k'-1}$ extends $B_k$, $B_{k'}$ also extends $B_k$.
    Alternatively, if no honest node used either of these rules to vote for $B_{k'}$ then all members of $H_1$ must have used the fallback vote rule to vote for $B_{k'}$. 
    Therefore, they must have received $\sig{\FallbackPropose, B_{k'}, \Cert_{v''}(B_{k''}), \TimeoutCert_{v'-1}, v'}$ such that $\Cert_{v''}(B_h)$ had a rank greater than or equal to that of the highest ranked block certificate in $\TimeoutCert_{v'-1}$ and $B_{k'}$ directly extended $B_{k''}$, before multicasting a $\TimeoutMessage$ message for $v'$ or any higher view.
    Moreover, by~\Cref{clm:honest-normal-proposal-cert-2}, $v'' < v'$.
    We now show that $v'' \ge v$.
    
    Recall that we know from the commit rule that $\Cert_{v+1}(B_{k+1})$ exists and that $B_{k+1}$ extends $B_k$.
    Therefore, a set $H_2$ of at least $f+1$ honest nodes must have voted for $B_{k+1}$ in view $v+1$. 
    By the vote rules and the lock rule, these nodes must have received $\Cert_v(B_k)$ and must not have sent $\TimeoutMessage_{v^*}$ with $v^* \ge v+1$ before doing so.
    Thus, by~\Cref{fact:timeout-qc}, every $\TimeoutMessage_{v^*}$ message sent by $H_2$ will necessarily contain $\Cert_v(B_k)$ or higher.
    Therefore, because $v'-1 \ge v+1$ and since by~\Cref{clm:honest-majority-intersection} $H_1$ and $H_2$ must have at least one node in common, the highest ranked block certificate of $\TimeoutCert_{v'-1}$ must have a rank of at least as great as $\Cert_v(B_k)$.
    Therefore, by extension and the definition of the rank of a block certificate given in~\Cref{sec:preliminaries}, $v'' \ge v$.
    Therefore, since $v \le v'' < v'$, by the induction hypothesis, $B_{k''}$ extends $B_k$, so $B_{k'}$ also extends $B_k$.
\end{proof}

\begin{theorem}[Safety]\label{thm:safety-2}
Honest nodes do not commit different values at the same log position.
\end{theorem}
The proof for~\Cref{thm:safety-2} remains the same as that given in~\Cref{thm:safety}, except that~\Cref{lem:uniq-xtn-pipelined} needs to be invoked instead of~\Cref{lem:uniq-xtn}.

\begin{claim}
    \label{clm:sequential-progress-2}
    If an honest node enters view $v$ then at least one honest node must have already entered $v-1$.
\end{claim}
\begin{proof}
    The view advancement rule requires an honest node to receive either $\Cert_{v-1}$ or $\TimeoutCert_{v-1}$ in order to enter $v$.
    Therefore, at least $f+1$ honest nodes must multicast the corresponding constituent messages.
    In the case of $\Cert_{v-1}$, the vote rules require these nodes to be in $v-1$ when they do so.
    In the case of $\TimeoutCert_{v-1}$, at least one honest node must have its view timer expire whilst in $v-1$ before any honest node can multicast $\TimeoutMessage_{v-1}$.
    Thus, in either case, an honest node can only enter $v$ if at least one honest node has already entered $v-1$.
\end{proof}

\begin{claim}
    \label{clm:timeout-bound}
    If the first honest node enters view $v$ at time $t$ then no honest node multicasts $\TimeoutMessage_{v'}$ for $v' \ge v$ before $t + 3\Delta$.
\end{claim}
\begin{proof}
    Since $t$ is defined as the time that the first honest node enters $v$, by the timeout rule, no honest node can have its view timer expire in $v$ before $t + 3\Delta$.
    Consequently, since $\network$ contains $f$ Byzantine nodes, at most $f$ $\TimeoutMessage_v$ messages can exist before this time.
    Therefore, since the timeout rule requires that a node either have its view timer expire whilst in $v$, or that it observe at least $f+1$ $\TimeoutMessage_v$ messages (since timeout certificates must be constructed from $2f+1$ such messages) before it may send $\TimeoutMessage_v$ itself, no honest node can send $\TimeoutMessage_v$ before $t + 3\Delta$.
    Moreover, by~\Cref{clm:sequential-progress-2}, no honest node can have entered $v'' > v$ before $t$.
    Thus, by the same argument, neither can any honest node send $\TimeoutMessage_{v''}$ before this time.
\end{proof}

\Cref{cor:tc-bound} follows from~\Cref{clm:timeout-bound} and the requirement that timeout certificates be constructed from $2f+1$ of timeout messages for the same view. 
\begin{corollary}
    \label{cor:tc-bound}
    If the first honest node enters view $v$ at time $t$ then $\TimeoutCert_{v'}$ cannot exist for $v' \ge v$ before $t + 3\Delta$.
\end{corollary}

\begin{lemma}
    \label{lem:view-sync-gst}
    Let $t_g$ denote GST. If the first honest node to enter view $v$, say $\node{i}$, does so at time $t$ such that $t \ge t_g$, then every honest node enters $v$ or higher before $t + 2\Delta$. 
\end{lemma}
\begin{proof}
    If any honest node enters view $v'$ such that $v' \ge v$ via $\Cert_{v'-1}$ before $t + \Delta$ then, by the view advancement rules, it will multicast $\Cert_{v'-1}$ and thus all honest nodes will enter $v'$ or higher before $t + 2\Delta$.
    Suppose, then, that no honest node enters $v'$ via $\Cert_{v'-1}$ before $t + \Delta$.
    Therefore, $\node{i}$ must enter $v$ via $\TimeoutCert_{v-1}$.
    Hence, at least $f+1$ honest nodes must multicast $\TimeoutMessage_{v-1}$ before $t$ and thus all will receive these messages before $t + \Delta$.
    Furthermore, since $t$ is defined as the time that the first honest node enters $v$, by~\Cref{cor:tc-bound}, $\TimeoutCert_{v'}$ cannot exist before $t + 3\Delta$ and thus all honest nodes must be in view $v$ or lower when they receive the aforementioned $\TimeoutMessage_{v-1}$ messages.
    Hence, by the timeout rule, all honest nodes in $v-1$ or lower that have not already multicasted $\TimeoutMessage_{v-1}$ will do so before $t + \Delta$.
    Moreover, every honest node that enters $v$ before this time must do so via $\TimeoutCert_{v-1}$ and thus, by the timeout rule, will also multicast $\TimeoutMessage_{v-1}$ before $t + \Delta$.
    Consequently, all honest nodes will multicast $\TimeoutMessage_{v-1}$ before this time, so they will all be able to construct $\TimeoutCert_{v-1}$ before $t + 2\Delta$.
    Thus, by the view transition rules, every honest node will enter $v$ or higher before $t + 2\Delta$. 
\end{proof}

\begin{lemma}
    \label{lem:view-sync}
    Let $t_g$ denote GST. If the first honest node to enter view $v$, say $\node{i}$, does so at time $t$, then every honest node enters $v$ or higher before $\max(t_g, t) + 3\Delta$. 
\end{lemma}
\begin{proof}
    If $t \ge t_g$ then~\Cref{lem:view-sync-gst} shows that all honest nodes will enter $v$ or higher before $\max(t_g, t) + 2\Delta$.
    Consider the case when $t < t_g$.
    Let $v''$ be the highest view of any honest node at $t_g$ and let $\node{j}$ be a node in $v''$ at this time.
    Observe that $v'' \ge v$.
    If any honest node enters a view higher than $v''$, say $v^*$, between $t_g$ and $t_g + \Delta$, then by~\Cref{lem:view-sync-gst} all honest nodes will enter $v^* > v$ or higher before $t_g + 3\Delta = \max(t_g, t) + 3\Delta$. 
    Otherwise, no honest node enters a view higher than $v''$ before $t_g + \Delta$. 
    In this case, if any honest node enters $v''$ via $\Cert_{v''-1}$ before $t_g + \Delta$ then all honest nodes will enter $v''$ before $t_g + 2\Delta < \max(t_g, t) + 3\Delta$.
    Otherwise, $\node{j}$ (and all other nodes in $v''$ at $t_g$) must have entered $v''$ via $\TimeoutCert_{v''-1}$ and hence would have multicasted $\TimeoutMessage_{v''-1}$ no later than the time that they did so.
    Moreover, at least $f+1$ honest nodes must have multicasted $\TimeoutMessage_{v''-1}$ before $t_g$ and thus honest nodes in views less than $v''$ will receive these messages before $t_g + \Delta$ and, by the timeout rule, will multicast $\TimeoutMessage_{v''-1}$ messages if they have not already done so.
    Thus, all honest nodes will multicast $\TimeoutMessage_{v''-1}$ before $t_g + \Delta$, so they will all be able to construct $\TimeoutCert_{v''-1}$ and enter $v''$ before $t_g + 2\Delta < \max(t_g, t) + 3\Delta$.
    Hence, in all cases, every honest node enters $v$ or higher before $\max(t_g, t) + 3\Delta$.
\end{proof}

\begin{lemma}\label{lem:view-progress-2}
    All honest nodes keep entering increasing views. 
\end{lemma}
The proof for~\Cref{lem:view-progress-2} remains the same as for~\Cref{lem:view-progress}, except that~\Cref{lem:view-sync} needs to be invoked instead of~\Cref{clm:all-enter}.

\begin{claim}
    \label{clm:honest-propose}
    If the first honest node to enter view $v$ does so at time $t$ after GST and $L_v$ is honest, then $L_v$ proposes before $t + \Delta$ and all honest nodes receive its proposal before $t + 2\Delta$.
\end{claim}
\begin{proof}
    Let $\node{i}$ be the first honest node to enter $v$.
    By the view advancement rule, $\node{i}$ may have entered $v$ via either $\Cert_{v-1}$ or $\TimeoutCert_{v-1}$.
    In the former case, it would have multicasted this certificate, and in the latter it would have unicasted it to $L_v$.
    Therefore, in either case, by the view advancement and proposal rules, $L_v$ will receive a certificate that will allow it to enter $v$ and propose before $t + \Delta$.
    Moreover, since $L_v$ is honest, it will multicast its proposal, so all honest nodes will receive it before $t + 2\Delta$.
\end{proof}

\begin{lemma}
    \label{lem:certification}
    If the first honest node to enter view $v$ does so at time $t$ after GST and $L_v$ is honest, then all honest nodes receive $\Cert_v(B_k)$ for some block $B_k$ proposed by $L_v$, before $t + 3\Delta$.
\end{lemma}
\begin{proof}
    By~\Cref{lem:view-safety-2}, only one block can become certified for a given view.
    Thus, if $\Cert_v(B_k)$ exists then any node that receives a view $v$ block certificate must receive $\Cert_v(B_k)$.
    Additionally, by~\Cref{lem:view-sync-gst} and~\Cref{clm:honest-propose}, all honest nodes enter $v$ or higher and receive a proposal from $L_v$ before $t + 2\Delta$.
    Moreover, since $t$ is defined as the time that the first honest node enters $v$, no honest node will multicast $\TimeoutMessage_v$ before $t + 3\Delta$.
    Therefore, if any honest node enters a view greater than $v$ before $t + 2\Delta$ then, by~\Cref{clm:sequential-progress-2}, at least one honest node must have already entered $v+1$ via $\Cert_v(B_k)$.
    By the view advancement rule, this node would have multicasted $\Cert_v(B_k)$, so all nodes will receive this certificate before $t + 3\Delta$, completing the proof.
    Alternatively, if no honest node receives $\Cert_v(B_k)$ before $t + 2\Delta$, then all honest nodes will enter $v$ before $t + 2\Delta$.
    Moreover, since $L_v$ is honest, it will ensure that its proposal is well-formed: i.e., if it is a normal proposal then the proposed block will extend the block certified by the included block certificate.
    Otherwise, if it is a fallback proposal then the proposed block will extend the block certified by $L_v$'s $\Lock_i$, which, by the lock rule, will have a rank at least as great as that of the block certificate with the highest rank in the included timeout certificate.
    Additionally, since $L_v$ is honest, it will create only one normal proposal or one fallback proposal. 
    Moreover, if it creates a normal proposal then any equivocal optimistic proposal that it may have created will necessarily have a different parent than the normal proposal because honest leaders propose fixed block payloads for a given view, per the definition of a block given in~\Cref{sec:preliminaries}.
    Consequently, by~\Cref{lem:view-safety-2}, the parent of the equivocal optimistic block proposal cannot be certified, so, by the optimistic vote rule, no honest node will be able to vote for this proposal.
    Finally, since all honest nodes will receive $L_v$'s proposal before $t + 2\Delta$, they cannot have $\TimeoutView_i \ge v$ by this time. Thus, by the vote rules, they will all vote for the included block, so all honest nodes will be able to construct $\Cert_v(B_k)$ before $t + 3\Delta$.
\end{proof}

\begin{lemma}
    \label{lem:locking-2}
    If the first honest node to enter view $v$ does so at time $t$ after GST and $L_v$ is honest, then at least $f+1$ honest nodes lock $\Cert_v(B_k)$ upon entering $v+1$, and enter $v+1$ without multicasting $\TimeoutMessage_{v'}$ for $v' \ge v$.
\end{lemma}
\begin{proof}
    By~\Cref{lem:certification}, all honest nodes will receive $\Cert_v(B_k)$ by $t + 3\Delta$. 
    Moreover, by~\Cref{cor:tc-bound}, $\TimeoutCert_{v'}$ cannot exist before $t + 3\Delta$. 
    Therefore, the only way for any honest node to exit view $v$ is via some $\Cert_{v'}$.
    We consider two cases: (i) when all honest nodes receive $\Cert_v(B_k)$ before $\Cert_{v''}$, where $v'' > v$, and; (ii) when at least one honest node does not. 
    In the first case, by the lock rule and~\Cref{lem:view-safety-2}, all honest nodes will have $\Lock_i < \Cert_v(B_k)$ before receiving $\Cert_v(B_k)$ and will therefore lock $\Cert_v(B_k)$ when they receive it.
    Moreover, since they cannot have received a certificate for $v'$ before $\Cert_v(B_k)$, they must be in view $v$ or lower when they receive this certificate and thus, by the view advancement rules, will enter $v+1$.
    Otherwise, at least one honest node must receive $\Cert_{v''}$ before $\Cert_v(B_k)$.
    In this case, by the definition of a block certificate, a set $H_1$ of at least $f+1$ honest nodes must vote towards $\Cert_{v''}$ after entering $v''$, which, as previously concluded, they must do via $\Cert_{v''-1}$.
    Implicitly then, a set $H_2$ of at least $f+1$ honest nodes must enter $v+1$ via $\Cert_v(B_k)$, and, since the view advancement rules therefore ensure that they cannot have received a certificate for a higher view before they do so, by the lock rule and the block certificate ranking rule, they will lock $\Cert_v(B_k)$.
    Finally, in both cases, because all honest nodes must receive $\Cert_v(B_k)$ before $t + 3\Delta$, by~\Cref{clm:timeout-bound}, no honest node can have multicasted $\TimeoutMessage_{v'}$ before exiting $v$.
\end{proof}

\begin{lemma}\label{lem:certified-successor-2}
    If the first honest node to enter view $v$ does so after GST, $L_v$ is honest and proposes a block $B_k$ that becomes certified, and $\Cert_{v+1}(B_{k'})$ exists, then $B_{k'}$ directly extends $B_k$.
\end{lemma}
\begin{proof}
    By~\Cref{lem:certification} and~\Cref{lem:locking-2}, all honest nodes will receive $\Cert_v(B_k)$ and a set $H$ of at least $f+1$ of them will lock it upon entering $v+1$ without multicasting $\TimeoutMessage_v$.
    Therefore, by~\Cref{clm:quorum-intersection} and~\Cref{clm:honest-majority-intersection}, $\TimeoutCert_v$ cannot exist, so $\Cert_{v+1}(B_{k'})$ must be $\Cert^o_{v+1}(B_{k'})$ or $\Cert^n_{v+1}(B_{k'})$.
    Additionally, by the same claims, at least one honest node, say $\node{i}$, must both lock $\Cert_v(B_k)$ and vote for $B_{k'}$.
    By the view advancement rule, $\node{i}$ would have entered a higher view than $v+1$ if it had received a block certificate for a higher view than $v$ before voting for $B_{k'}$.
    Thus, since the vote rules only allow $\node{i}$ to vote towards $\Cert_{v+1}(B_{k'})$ whilst in $v+1$ and because it must lock $\Cert_v(B_k)$ upon entering $v+1$, it must have been locked on $\Cert_v(B_k)$ when it voted for $B_{k'}$.
    Furthermore, since the optimistic vote rule requires $\node{i}$ to be locked on the parent of $B_{k'}$, if $\node{i}$ votes for an optimistic proposal containing $B_{k'}$ then $B_{k'}$ must directly extend $B_k$.
    Similarly, the normal vote rule requires the proposal containing $B_{k'}$ to be justified by some $\Cert_v$ that certifies the parent of $B_{k'}$.
    By~\Cref{lem:view-safety-2}, $\Cert_v = \Cert_v(B_k)$, so $B_{k'}$ must directly extend $B_k$.
\end{proof}

\begin{theorem}[Liveness]\label{the:liveness-2}
Each client request is eventually committed by all honest nodes.
\end{theorem}
\begin{proof}
As in~\Cref{thm:liveness}, we show that all honest nodes continue to commit new blocks to their local blockchains after GST, which, together with the assumptions mentioned along with~\Cref{def:smr} in~\Cref{sec:preliminaries}, is sufficient to achieve liveness.

By~\Cref{lem:view-progress-2}, all honest nodes continually enter higher views.
Therefore, the protocol eventually reaches two consecutive views after GST, say $v$ and $v+1$, that have leaders $L_v$ and $L_{v+1}$ that are both honest.
By~\Cref{lem:certification} and~\Cref{lem:locking-2}, all honest nodes will receive $\Cert_v(B_k)$ and at least $f+1$ of them will lock it.
By the same lemmas, the same is also true for $\Cert_{v+1}(B_{k'})$.
Moreover, by~\Cref{lem:certified-successor-2}, $B_{k'}$ directly extends $B_k$; i.e., $k' = k+1$.
Consequently, by the commit rule, all honest nodes will commit $B_k$ upon receiving both $\Cert_v(B_k)$ and $\Cert_{v+1}(B_{k+1})$.
Hence, \pipelined commits a new block every time two consecutive, honest leaders are elected after GST.
\end{proof}

\begin{theorem}[Reorg resilience]\label{the:reorg-resilience-2}
If the first honest node to enter view $v$ does so after GST and $L_v$ is honest and proposes, then one of its proposed blocks, say $B_k$, becomes certified and for every $\Cert_{v'}(B_{k'})$ such that $v' \ge v$, $B_{k'}$ extends $B_k$.
\end{theorem}
\begin{proof}
By~\Cref{lem:certification}, $L_v$ produces a certified block. 
Let this block be denoted $B_k$.
We now show that for every $\Cert_{v'}(B_{k'})$, $B_{k'}$ extends $B_k$.

If $v' = v$ then, by~\Cref{lem:view-safety-2}, $B_{k'} = B_k$ and thus, per the definition of block extension given in~\Cref{sec:preliminaries}, $B_{k'}$ extends $B_k$. We now complete the proof for $v' > v$ by strong induction on $v'$, however, since~\Cref{lem:certified-successor-2} covers the base case ($v' = v+1$), we proceed directly with the inductive step.

\paragraph{Inductive step: $v' > v+1$.} For our induction hypothesis, we assume that the theorem holds up to view $v'-1$. 
That is, we assume that every $\Cert_{v^*}(B_{k^*})$ with $v \le v^* < v'$ extends $B_k$.
We use this assumption to prove that it also holds for $v'$.
We first observe that the existence of $\Cert_{v'}(B_{k'})$ implies that a set $H_1$ of at least $f+1$ honest nodes voted for $B_{k'}$ in view $v'$.
If any of these nodes voted using the optimistic or normal vote rules then they must have received $\Cert_{v'-1}(B_{k'-1})$ and $B_{k'}$ must extend $B_{k'-1}$.
Therefore, since by the induction hypothesis $B_{k'-1}$ extends $B_k$, $B_{k'}$ also extends $B_k$.
Alternatively, if no honest node used either of these rules to vote for $B_{k'}$ then all members of $H_1$ must have used the fallback vote rule to vote for $B_{k'}$. 
Therefore, they must have received $\sig{\FallbackPropose, B_{k'}, \Cert_{v''}(B_{k''}), \TimeoutCert_{v'-1}, v'}$ such that $\Cert_{v''}(B_h)$ had a rank greater than or equal to that of the highest ranked block certificate in $\TimeoutCert_{v'-1}$ and $B_{k'}$ directly extended $B_{k''}$, before multicasting a $\TimeoutMessage$ message for $v'$ or any higher view.
Moreover, by~\Cref{clm:honest-normal-proposal-cert-2}, $v'' < v'$.
We now show that $v'' \ge v$.

By the fallback vote rule, $H_1$ will not vote for a fallback proposal for $v'$ unless it contains a valid $\TimeoutCert_{v'-1}$.
Hence, a set $H_2$ of at least $f+1$ honest nodes must multicast $\TimeoutMessage_{v'-1}$.
Moreover, by~\Cref{lem:locking-2}, a set $H_3$ of at least $f+1$ honest nodes must lock $\Cert_v(B_k)$ while entering $v+1$ and must do so without having multicasted $\TimeoutMessage_{v'}$ for $v' \ge v$.
By~\Cref{clm:honest-majority-intersection}, $H_2$ and $H_3$ must have at least one node, say $\node{i}$, in common.
Therefore, since $\node{i}$ must lock $\Cert_v(B_k)$ before sending $\TimeoutMessage_{v'-1}$, by the lock rule and the block certificate ranking rule, every valid $\TimeoutCert_{v'-1}$ must contain a block certificate with a rank of at least $v$.
Thus, $\Cert_{v''}(B_{k''}) \ge \Cert_v(B_k)$, so $v'' \ge v$.
Therefore, since $v \le v'' < v'$, by the induction hypothesis, $B_{k'}$ extends $B_k$.
\end{proof}

\section{\commit Security Analysis}
\label{sec:cm_proof}
As previously observed, because \commit retains the vote and commit rules of \pipelined, the same liveness argument that can be made for the latter also applies to the former. The same is also true for the reorg resilience of the protocol. However, \commit's new commit path requires additional justification. We prove the safety of the additional rules below, and observe that~\Cref{thm:safety-2} holds for \commit when~\Cref{lem:uniq-xtn-commit} is invoked along with~\Cref{lem:uniq-xtn-pipelined}. We subsequently prove that \commit requires only a single honest leader to commit a new block after GST.

\begin{claim}\label{clm:commit-qc-no-tc}
    If $2f+1$ distinct nodes send $\sig{\Commit, H(B_k), v}_*$ then $\TimeoutCert_v$ cannot exist.
\end{claim}
\begin{proof}
    Let $H$ denote the set of $f+1$ honest nodes that sent $\sig{\Commit, H(B_k), v}_*$.
    By the pre-commit rules, any member of $H$ that sent this message must have received $\Cert_v(B_k)$ whilst having $\TimeoutView_* < v$.
    Therefore, by the view advancement rule, all members of $H$ must have entered $v+1$ before sending $\Timeout_v$ and thus, by the timeout rule, will never send $\Timeout_v$.
    Therefore, since this implies that at most $2f$ distinct $\Timeout_v$ messages will ever exist, $\TimeoutCert_v$ cannot exist.
\end{proof}

\begin{claim}\label{clm:commit-qc-future-tcs}
    If $2f+1$ distinct nodes send $\sig{\Commit, H(B_k), v}_*$ then every $\TimeoutCert_{v'}$ for view $v' > v$ must contain a block certificate for $v$ or higher.
\end{claim}
\begin{proof}
    Let $H$ denote the set of $f+1$ honest nodes that sent $\sig{\Commit, H(B_k), v}_*$.
    By the lock rule, the members of $H$ that voted to commit $B_k$ via the direct pre-commit rule must have locked $\Cert_v(B_k)$ upon receiving it.
    Comparatively, those that voted to commit $B_k$ via the indirect pre-commit rule must have previously sent $\sig{\Commit, H(B_l), v''}_*$ for some descendant $B_l$ of $B_k$ for some $v'' > v$.
    Thus, by the pre-commit and lock rules, these nodes must have received and locked $\Cert_{v''}(B_l)$ or some higher ranked block certificate before sending $\sig{\Commit, H(B_k), v}_*$.
    In either case, the members of $H$ must have locked a block certificate for $v$ or higher whilst having $\TimeoutView_* < v$, so all timeout messages sent by these nodes for views greater than $v$ will necessarily contain $\Cert_v(B_k)$ or higher.
    Thus, since every $\TimeoutCert_{v'}$ for view $v' > v$ must contain a timeout message from at least one member of $H$, every such certificate must contain a block certificate for $v$ or higher.
\end{proof}

\begin{lemma}[Unique Extensibility Continued]\label{lem:uniq-xtn-commit}
    If an honest node, say $\node{i}$, directly commits a block $B_k$ via the alternative direct commit rule that was certified for view $v$ and $\Cert_{v'}(B_{k'})$ exists such that $v' \ge v$, then $B_{k'}$ extends $B_k$.
\end{lemma}
\begin{proof}
    If $v' = v$ then, by~\Cref{lem:view-safety-2}, $B_{k'} = B_k$ and thus, per the definition of block extension given in~\Cref{sec:preliminaries}, $B_{k'}$ extends $B_k$.
    We now complete the proof for $v' > v$ by strong induction on $v'$.

    \paragraph{Base case: $v' = v+1$.} By~\Cref{clm:commit-qc-no-tc}, $\TimeoutCert_v$ cannot exist.
    Therefore, $\Cert_{v'}(B_{k'})$ must be either $\Cert^o_{v+1}(B_{k'})$ or $\Cert^n_{v+1}(B_{k'})$.
    In either case, by the respective vote rules, $B_{k'}$ extends $B_k$.

    \paragraph{Inductive step: $v' > v+1$.} For our induction hypothesis, we assume that the lemma holds up to view $v'-1$. 
    That is, we assume that every $\Cert_{v^*}(B_{k^*})$ with $v \le v^* < v'$ extends $B_k$.
    We use this assumption to prove that it also holds for $v'$. 
    We first observe that the existence of $\Cert_{v'}(B_{k'})$ implies that a set $H_1$ of at least $f+1$ honest nodes voted for $B_{k'}$ in view $v'$.
    If any of these nodes voted using the optimistic or normal vote rules then they must have observed $\Cert_{v'-1}(B_{k'-1})$ and $B_{k'}$ must extend $B_{k'-1}$.
    Therefore, since by the induction hypothesis $B_{k'-1}$ extends $B_k$, $B_{k'}$ also extends $B_k$.
    Alternatively, if no honest node used either of these rules to vote for $B_{k'}$ then all members of $H_1$ must have used the fallback vote rule to vote for $B_{k'}$. 
    Therefore, they must have received $\sig{\FallbackPropose, B_{k'}, \Cert_{v''}(B_{k''}), \TimeoutCert_{v'-1}, v'}$ such that $\Cert_{v''}(B_{k''})$ had a rank at least as great as that of the highest ranked block certificate in $\TimeoutCert_{v'-1}$ and $B_{k'}$ directly extended $B_{k''}$, before multicasting $\TimeoutMessage$ for $v'$ or any higher view.
    Moreover, by~\Cref{clm:honest-normal-proposal-cert-2}, $v'' < v'$.
    We now show that $v'' \ge v$.
    
    By the alternative direct commit rule, $\node{i}$ must have received $2f+1$ distinct $\sig{\Commit, H(B_k), v}_*$, at least $f+1$ of which must have been sent by a set $H_2$ of distinct honest nodes.
    Therefore, by~\Cref{clm:commit-qc-future-tcs}, every timeout certificate for a view greater than $v$ must contain a block certificate for $v$ or higher.
    Thus, because $v'-1 \ge v+1$ and since by~\Cref{clm:honest-majority-intersection} $H_1$ and $H_2$ must have at least one node in common, the highest ranked block certificate of $\TimeoutCert_{v'-1}$ must have a rank of at least as great as $\Cert_v(B_k)$.
    Hence, $v'' \ge v$.
    Therefore, since $v \le v'' < v'$, by the induction hypothesis, $B_{k''}$ extends $B_k$, so $B_{k'}$ also extends $B_k$.
\end{proof}

\begin{claim}[Single Leader Commit]\label{clm:single-leader-commit}
    If the first honest node to enter view $v$ does so at time $t$ after GST and $L_v$ is honest, then all honest nodes commit one of its proposals before $t + 4\Delta$.
\end{claim}
\begin{proof}
    By~\Cref{lem:certification}, all honest nodes receive $\Cert_v(B_k)$ for some block $B_k$ proposed by $L_v$, before $t + 3\Delta$.
    Therefore, by the alternative direct commit rule, if they all multicast $\sig{\Commit, H(B_k), v}_*$ upon receiving this certificate, then all honest nodes will commit $B_k$ before $t + 4\Delta$.
    Otherwise, at least one honest node, say $\node{i}$, must fail to send $\sig{\Commit, H(B_k), v}_*$ before $t + 3\Delta$.
    However, by~\Cref{clm:timeout-bound}, no honest node can have its view timer expire before this time, so $\node{i}$ cannot have $\TimeoutView_i \ge v$ upon receiving $\Cert_v(B_k)$.
    Thus, by the pre-commit rules, $\node{i}$ must neither be in view $v$ or lower, nor have multicasted a commit vote for any descendant of $B_k$, upon receiving $\Cert_v(B_k)$.
    Let $v'$ denote the view that $\node{i}$ was in upon receiving this certificate.
    Since, by~\Cref{cor:tc-bound}, no timeout certificate can exist for $v$ or higher before $t + 3\Delta$, $\node{i}$ must have entered $v'$ via $\Cert_{v'-1}(B_l)$.
    Moreover, by the direct pre-commit rule, it would have multicasted $\sig{\Commit, H(B_l), v'-1}_j$ upon receiving this certificate.
    Therefore, $B_l$ must not be a descendant of $B_k$.
    However, it must also be true that $v' > v + 1$ and that all block certificates for the views between $v'$ and $v$ must be either $\Cert^o$ or $\Cert^n$.
    Therefore, by the corresponding vote rules, $B_l$ must be a descendant of $B_k$, contradicting the former conclusion that it must not be.
    Therefore, all honest nodes will multicast $\sig{\Commit, H(B_k), v}_*$ before $t + 3\Delta$, so all honest nodes will commit $B_k$ before $t + 4\Delta$.
\end{proof}

\section{Further Optimizing View Length}\label{sec:view_length_optimiztion}
The view lengths of \pipelined and \commit can be further reduced in views without valid proposals by applying the modifications presented in~\Cref{fig:view-length-optimization}. These modifications leverage the fact that these protocols guarantee that all nodes will receive a valid proposal from an honest $L_v$ within $2\Delta$ of the first honest node entering $v$ after GST (see~\Cref{clm:honest-propose}). Consequently, if a node has to wait more than $2\Delta$ to receive a valid proposal, it can be confident either that the network is asynchronous or that the leader is Byzantine and therefore has reason to begin the view change process. Similarly, if it votes for such a proposal then it should expect to construct a certificate for the included block within $3\Delta$ of the first honest node entering the view (see~\Cref{lem:certification}) and thus should increase its view timer by $\Delta$ upon voting to allow sufficient time for the votes of its peers to arrive. While this latter modification preserves liveness, it means that \pipelined and \commit will only exhibit a view length of $2\Delta$ in views in which no honest nodes vote for any block. Accordingly, it is trivial for Byzantine leaders to ensure that their views retain the original view length of $3\Delta$. Even so, this modification represents a meaningful optimization in the crash fault tolerant setting.

\begin{figure}[!t]
\small
    \begin{boxedminipage}[ht]{\linewidth}
    The \vl of \pipelined and \commit can be further reduced by modifying the protocol for for $\node{i}$ presented in~\Cref{fig:pipelined-moonshot} as follows:
    \begin{enumerate}[leftmargin=*]
        \item \textbf{Reset Timer.} Upon entering $v$, reset $\ViewTimer_i$ to $2\Delta$ and start counting down. This replaces the corresponding logic given in the \textbf{Advance View} rule.
        \item \textbf{Increase Timer.} Upon voting in $v$, increase $\ViewTimer_i$ by $\Delta$.
    \end{enumerate}
\end{boxedminipage}
\caption{Further Optimizing View Length}
\label{fig:view-length-optimization}
\end{figure}

\end{document}